\documentclass[journal,12pt,onecolumn,draftclsnofoot]{IEEEtran}

\usepackage{amsmath} 
\usepackage{amsthm}
\usepackage{amssymb}
\usepackage{graphicx}
\usepackage{setspace}
\usepackage{soul}

\makeatletter
\theoremstyle{plain}
\newtheorem{thm}{\protect\theoremname}
\theoremstyle{plain}
\newtheorem{lem}[thm]{\protect\lemmaname}
\newtheorem*{remark}{\protect\remarkname}

\makeatother

\providecommand{\lemmaname}{Lemma}
\providecommand{\theoremname}{Theorem}
\providecommand{\remarkname}{Remark - Practical Implementation Issues}

\newcommand\blfootnote[1]{%
  \begingroup
  \renewcommand\thefootnote{}\footnote{#1}%
  \addtocounter{footnote}{-1}%
  \endgroup
}

\begin{document}

\pagestyle{empty}

\title{Non-Orthogonal eMBB-URLLC Radio Access for Cloud Radio Access Networks with Analog Fronthauling}
\author{Andrea Matera, Rahif Kassab, Osvaldo Simeone and Umberto Spagnolini}

\maketitle

\thispagestyle{empty}

\begin{abstract}
This paper considers the coexistence of Ultra Reliable Low Latency
Communications (URLLC) and enhanced Mobile BroadBand (eMBB) services
in the uplink of Cloud Radio Access Network (C-RAN) architecture based
on the relaying of radio signals over analog fronthaul links. While
Orthogonal Multiple Access (OMA) to the radio resources enables the
isolation and the separate design of different 5G services, Non-Orthogonal
Multiple Access (NOMA) can enhance the system performance by sharing
wireless and fronthaul resources. This paper provides an information-theoretic
perspective in the performance of URLLC and eMBB traffic under both
OMA and NOMA. The analysis focuses on standard cellular models with
additive Gaussian noise links and a finite inter-cell interference
span, and it accounts for different decoding strategies such as puncturing, Treating Interference as Noise (TIN) and Successive Interference Cancellation
(SIC). Numerical results demonstrate that, for the considered analog
fronthauling C-RAN architecture, NOMA achieves higher eMBB rates with
respect to OMA, while guaranteeing reliable low-rate URLLC communication
with minimal access latency. Moreover, NOMA under SIC is seen to achieve
the best performance, while, unlike the case with digital capacity-constrained
fronthaul links, TIN always outperforms puncturing.
\end{abstract}

\begin{IEEEkeywords}
network slicing, RoC, URLLC, eMBB, C-RAN.
\end{IEEEkeywords}

\blfootnote{Andrea Matera and Umberto Spagnolini are with Dipartimento di Elettronica, Informazione e Bioingegneria (DEIB), Politecnico di Milano, 20133 Milano, Italy; {andrea.matera,umberto.spagnolini}@polimi.it.\\
Rahif Kassab and Osvaldo Simeone are with Centre for Telecommunications Research (CTR), Department of Informatics, King\textquoteright s College London, London WC2B 4BG, UK; {rahif.kassab,osvaldo.simeone}@kcl.ac.uk.
\\ \\This work has received funding from the European Research Council (ERC) under the European
Union Horizon 2020 Research and Innovation Programme (Grant Agreement No. 725731).} 

\section{Introduction}

Accommodating the heterogeneity of users' requirements is one of the
main challenges that both industry and academia are facing in order
to make 5G a reality \cite{Shafi-Molish-etal_2017}. In fact, next-generation
wireless communication systems must be designed to provision
different services, each of which with distinct constraints in terms
of latency, reliability, and information rate. In particular, 5G is
expected to support three different macro-categories of services,
namely enhanced Mobile BroadBand (eMBB), massive Machine-Type Communications
(mMTC), and Ultra-Reliable and Low-Latency Communications (URLLC)
\cite{5g2016view,MinimumRequierementsITU,3GPP_TR38802}. 

eMBB service is meant to provide very-high data-rate communications
as compared with current (4G) networks. This can be generally achieved
by using codewords that spread over a large number of time-frequency
resources, given that latency is not an issue. mMTC supports low-rate
bursty communication between a massive number of uncoordinated devices and the network. Finally, URLLC is designed to ensure low-rate
ultra-reliable radio access for a few nodes, while guaranteeing very
low-latency. As a result, URLLC transmissions need to be localized
in time, and hence URLLC packets should be short \cite{Durisi-Koch-Popovski_2016}.

The coexistence among eMBB, mMTC and URLLC traffic types can be ensured
by slicing the Radio Access Network (RAN) resources into non-overlapping,
or orthogonal, blocks, and by assigning distinct resources to different
services. With the resulting Orthogonal Multiple Access (OMA), the
target quality-of-service guarantees can be achieved by designing
each service separately \cite{Zhang-Liu_etal_2017,Popovski-Nielsen_etal_2018}.
However, when URLLC or mMTC traffic types are characterized by short
and bursty transmissions at random time instants, resources allocated
statically to these services are likely to be unused for most of time,
and thus wasted. A more efficient use of radio resources can
be accomplished by Non-Orthogonal Multiple Access (or NOMA), which
allows multiple services to share the same physical resources. 

By enabling an opportunistic shared use of the radio resources, NOMA
can provide significant benefits in terms of spectrum efficiency,
but it also poses the challenge of designing the system so that the
heterogeneous requirements of the services are satisfied despite the
mutual interference. The objective of this paper is to address this
issue by considering a Cloud-RAN (C-RAN) architecture characterized
by analog fronthaul links, referred to as Analog Radio-over-X, which
is introduced in the next section.

\subsection{C-RAN Based on Analog Radio-over-X Fronthauling }

The advent of 5G is introducing advanced physical layer technologies
and network deployment strategies such as massive MIMO, mmWave,
small-cell densification, mobile edge computing, \linebreak etc. (see \cite{Boccardi-Heath-etal_2014,wong_schober_ng_wang_2017} for an overview). C-RAN is
an enabling technology that is based on the colocation of the Base
Band Unit (BBU) of Edge Nodes (ENs) that are densely distributed in
a given indoor or outdoor area. This solution has the advantages of
allowing for centralized BBU signal processing, providing network
scalability, increasing spectral efficiency, and reducing costs.

The most typical C-RAN architecture relies on digital optical fronthaul
links to connect ENs to BBUs. This solution, known as Digital Radio-over-Fiber
(D-RoF), is adopted in current 4G mobile networks, and is based on
the transmission of in-phase and quadrature baseband signals, upon
digitization and packetization according to the CPRI protocol \cite{CPRI:specs}. 

Over the last years, several alternative C-RAN architectures have
been proposed that redistribute the RAN functionalities between BBU
and ENs, obtaining different trade-offs in terms of bandwidth and
latency requirements, advanced Cooperative Multi-Point processing
capabilities, and EN cost and complexity \cite{Bartelt:FunctSplit}.
For scenarios with stringent cost and latency constraints, a promising
solution is to use analog fronthauling.

With analog fronthauling, focusing on the uplink, the ENs directly
relay the radio signals to the BBUs after frequency translation and,
possibly, signal amplification. This has the advantages of avoiding
any bandwidth expansion due to digitization; guaranteeing ENs synchronization;
minimizing latency; reducing hardware cost; and improving energy efficiency
\cite{wake:RoFdesign,gambini2010radio,Gambini-Spagnolini_2013}. A
C-RAN architecture based on analog fronthauling is also known in the
literature as Analog Radio-over-X (A-RoX), where X depends on the
technology employed for the fronthaul, which can be either Fiber (A-RoF
\cite{wake:RoFdesign}), Radio (A-RoR \cite{bartelt2013radio}), or
Copper (A-RoC \cite{gambini2010radio}), as depicted in Figure \ref{fig:C-RAN-architecture-overview:}.
\begin{figure}
\begin{centering}
\includegraphics[width=0.5\paperwidth]{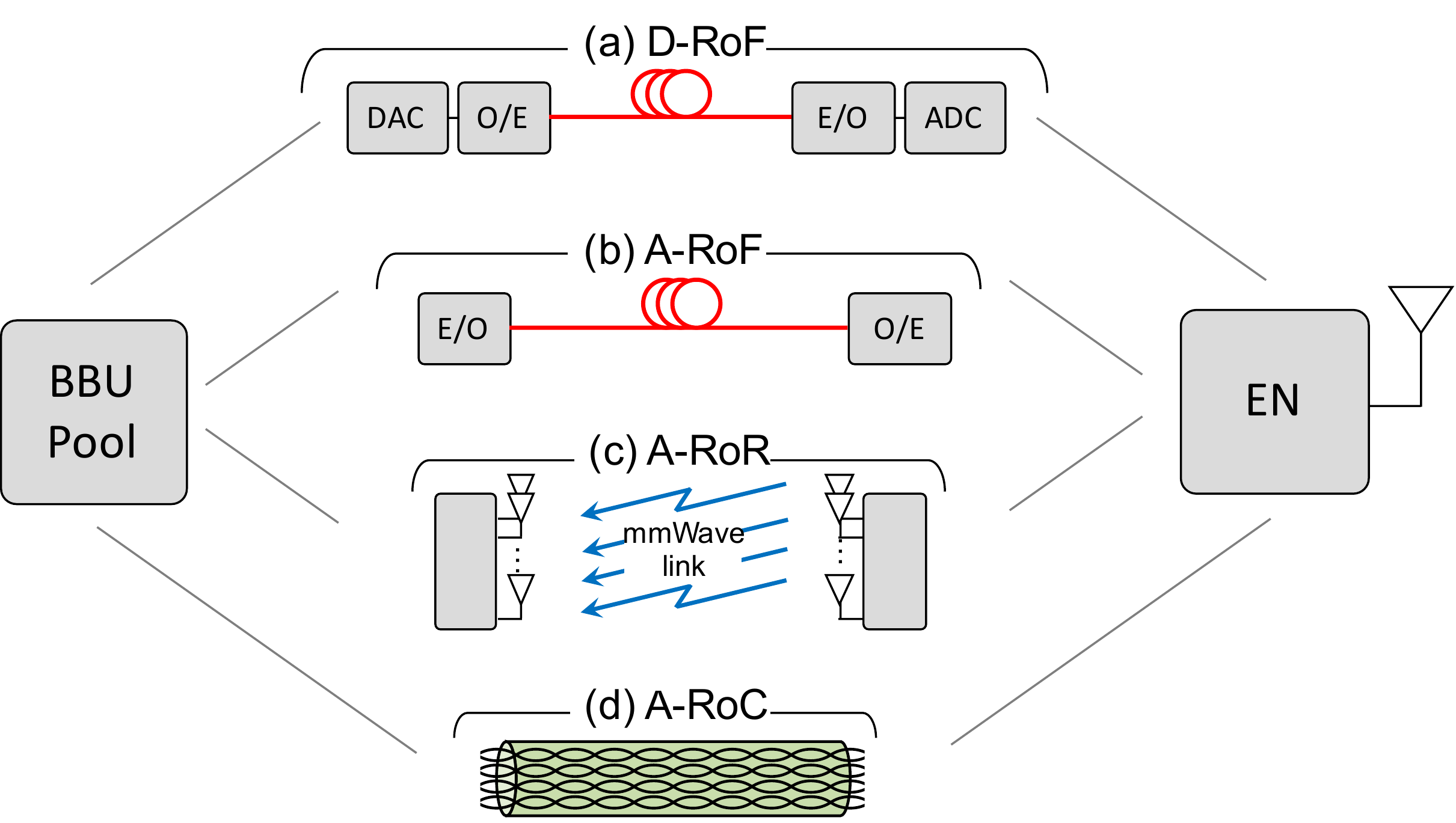}
\par\end{centering}
\caption{\label{fig:C-RAN-architecture-overview:}C-RAN architecture overview
for uplink direction: ({\bf a}) Digital Radio-over-Fiber, ({\bf b}) Analog Radio-over-Fiber,
({\bf c}) Analog Radio-over-Radio, ({\bf d}) Analog Radio-over-Copper.}
\end{figure}

In particular, A-RoF provides an effective example of analog fronthauling,
due to its capability to support the transport of large bandwidths
\cite{wake:RoFdesign,PWMglobecom}. However, A-RoF requires the deployment
of a fiber optic infrastructure whose installation is not always feasible,
e.g., in dense urban areas. In such scenarios, a possible solution
is to rely on the A-RoR concept, thus employing point-to-point wireless
links, mainly based on mmWave or THz bands, with several advantages
in terms of flexibility, resiliency, hardware complexity and cost
\cite{bartelt2013radio,Pham-Atsushi-etal_2016}. Another application
scenario where the installation of fiber links may be too expensive
to provide satisfactory business cases is indoor coverage. For indoor
deployments, A-RoC \cite{matera2017optimal,Matera-Combi-etal_2017,Matera-Spagnolini_2018}
has been recently proved to be an attractive solution, especially
from the deployment costs perspective \cite{Tonini-Fiorani-etal_2017},
since it leverages the pre-existing Local Area Network (LAN) cabling
infrastructure of building and enterprises. Moreover, LAN cables are
equipped with four twisted-pairs with a transport capability up to
500 MHz each, or 2 GHz overall, for radio signals, thus providing
enough bandwidth for analog fronthaul applications \cite{Naqvi-Matera-etal_2017}. Over the last years, A-RoC based on LAN cables has become a standard solution for in-building commercial C-RAN deployments, allowing to extend the indoor coverage over distances longer than 100 m \cite{RadioDots2014}.

In this paper we study the coexistence of URLLC and eMBB services
under both OMA and NOMA assuming a C-RAN multi-cell architecture based
on analog fronthauling (see Figure \ref{fig:C-RAN-architecture-overview:})
by using information theoretical tools.

\subsection{Related Works}

The concept of NOMA is well-known
from the information theoretic literature \cite{cover2012elements},
and its application to 5G dates back to \cite{Saito-Kishiyama-etal_2013},
where authors demonstrated for a single-cell scenario that superimposing
multiple users in the same resources achieves superior performance
with respect to conventional LTE networks, provided that the resulting
interference is properly taken care of. The extension of NOMA to multi-cell
networks is presented in \cite{Shin-Mojtaba-etal_2017}, which addresses
several multi-cell NOMA challenges, including coordinated scheduling,
beamforming, and practical implementation issues related to successive
interference cancellation. 

In contrast with most of the works on NOMA, which deal with homogeneous
traffic conditions (see \cite{Ding-Lei-etal_2017} for a recent review),
here the focus is on NOMA techniques in the context of heterogeneous
networks, such as the forthcoming 5G wireless systems, as discussed
in \cite{Anand_etal_2017,popovski20185g,vaezi2018multiple}. In fact, NOMA represents
an attractive solution to meet the distinct requirements of 5G services,
as it improves spectral efficiency (eMBB), enables massive device
connectivity (mMTC), and allows for low-transmission latency (URLLC)
\cite{Dai-Wang-etal_2015}. 

In \cite{popovski20185g} a communication-theoretic model was introduced
to investigate the performance trade-offs for eMBB, mMTC and URLLC
services in a single-cell scenario under both OMA and NOMA. This single-cell
model has been later extended in \cite{Kassab} to the uplink of a
multi-cell C-RAN architecture, in which the BBU communicates with
multiple URLLC and eMBB users belonging to different cells through
geographically distributed ENs. In the C-RAN system studied in \cite{Kassab},
while the URLLC signals are locally decoded at the ENs due to latency
constraints, the eMBB signals are quantized and forwarded over limited-capacity
digital fronthaul links to the BBU, where centralized joint decoding
\linebreak is performed. 

None of the aforementioned works considers the coexistence of different
5G services in a C-RAN architecture based on analog fronthaul links
which is the focus of this paper.

\subsection{Contributions}

In this paper we study for the first time the coexistence between
URLLC and eMBB services in the uplink of a C-RAN system with analog
fronthauling in which the URLLC signals are still decoded locally
at the EN, while the eMBB signals are forwarded to the BBU over analog
fronthaul links. 

In particular, the main contributions of this paper are three-fold: 
\begin{itemize}
\item We extend the uplink C-RAN theoretic model proposed in \cite{Kassab}
to the case of analog fronthaul, assuming that the fronthaul links
are characterized by multiple, generally interfering, channels that
carry the received radio signals;
\item By leveraging information theoretical tools, we investigate the performance
trade-offs between URLLC and eMBB services under both OMA and NOMA,
by considering different interference management strategies such as
puncturing, considered for the standardization
of 5G New Radio \cite{qualcomm_puncturing,3GPP_FinalReport_Puncturing},
Treating Interference as Noise (TIN), and Successive Interference
Cancellation~(SIC);
\item The analysis demonstrates that NOMA allows for higher eMBB information
rates with respect to OMA, while guaranteeing a reliable low-rate
URLLC communication with minimal access latency. Moreover, differently
from the case of conventional digital C-RAN architecture based on
limited-capacity fronthaul links \cite{Kassab}, in analog C-RAN,
TIN always outperforms puncturing, while the best performance is still
achieved by NOMA with SIC.
\end{itemize}

\subsection{Organization}

The remainder of the paper is organized as follows. The considered
system model is introduced in Section \ref{sec:System-Model}. Section \ref{sec:Cable-Fronthaul-Signal}
details the fronthaul signal processing techniques employed to cope
with the impairments of the A-RoC fronthaul links. The eMBB rand URLLC
information rates are discussed in Section \ref{sec:Orthogonal-Multiple-Access}
and Section \ref{sec:Non-Orthogonal-Multiple-Access} for OMA and NOMA,
respectively. Numerical results are presented in Section \ref{sec:Numerical-Results},
and Section \ref{sec:Conclusions} concludes the paper.

\subsection{Notation\label{subsec:Notation}}

Bold upper- and lower-case letters describe matrices and column vectors, respectively.
Letters $\mathbb{R}$, and $\mathbb{C}$ refer to real and complex
numbers, respectively. We denote matrix inversion, transposition and
conjugate transposition as $\left(\cdot\right)^{-1},\left(\cdot\right)^{T},\left(\cdot\right)^{H}$.
Matrix $\mathbf{I}_{n}$ is an identity matrix of size $n$ and $\mathbf{1}_{n}$
is a column vector made by $n$ ``1s''. Symbol $\otimes$ denotes
the Kronecker operator, $\text{vec}\cdot)$ is the vectorization
operator, and $\text{\ensuremath{\mathbb{E}}}[\cdot]$ is the statistical
expectation. Notation $\text{diag}(A_{1,}A_{2},...A_{n})$ denotes
a diagonal matrix with elements $A_{1,}A_{2},...A_{n}$ on the main
diagonal. The Q-function $\mathcal{Q}(\cdot)$ is the complementary
cumulative distribution function of the standardized normal random
variable, and $\mathcal{Q}^{-1}(\cdot)$ is its inverse.

\section{System Model\label{sec:System-Model}}

The C-RAN architecture under study is illustrated in Figure \ref{fig:Uplink-RoC-based}.
In this system, the BBUs communicate with multiple user equipments
(UEs) belonging to $M$ cells through $M$ single-antenna Edge Nodes
(ENs). The BBUs are co-located in the so-called BBU pool so that joint
processing can be performed, while the ENs are geographically distributed.
In particular, we assume here that cells are arranged in a line following
the conventional circulant Wyner model (\cite{simeone2012cooperative}, Chapter 2),
and each cell contains two single-antenna UEs with different service
constraints: one eMBB user and one URLLC user. 
\begin{figure}
\begin{centering}
\includegraphics[width=0.5\paperwidth]{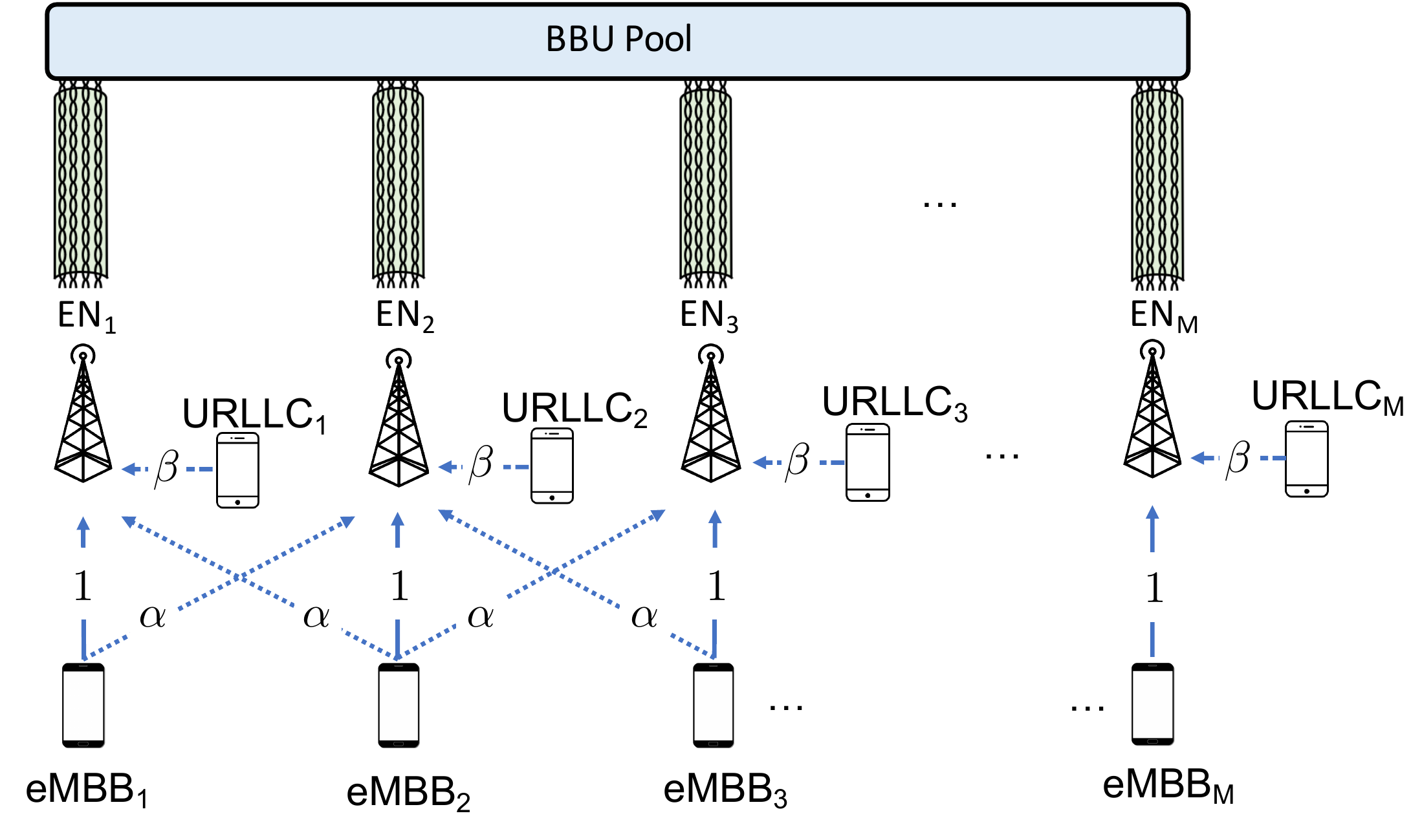}
\par\end{centering}
\caption{\label{fig:Uplink-RoC-based}Model of the uplink of C-RAN system based
on Analog Radio-over-Copper (A-RoC) fronthauling.}
\end{figure}

Due to the strict latency constraints of URLLC traffic, the signal
for the URLLC UEs is decoded on-site at the EN, while the eMBB signals
are forwarded to the BBU through a multi-channel analog fronthaul. In this hybrid cloud-edge architecture, the mobile operator equips the EN with edge computing capabilities in order to provision the services required by the URLLC user directly from the EN.
Following the A-RoX concept, the end-to-end channel from the eMBB
UEs and the BBU pool is assumed to be fully analog: the EN performs
only signal amplification and frequency translation to comply with
fronthaul capabilities and forwards the signals to the BBU, where
centralized decoding is performed. In practice, as detailed later in the paper (see Section \ref{sec:Orthogonal-Multiple-Access} and Section \ref{sec:Non-Orthogonal-Multiple-Access}), we assume that each EN hosts a digital module, responsible for URLLC signal decoding, and an analog module, responsible for the mapping of received radio signal over the analog fronthauling, which is identified as Analog-to-Analog (A/A) mapping.

While the technology used for the analog fronthauling can be either
fiber-optics (A-RoF), wireless (A-RoR) or cable (A-RoC), the system
model proposed in this paper reflects mainly the last two solutions.
Furthermore, we will focus on A-RoC, and we will
use the corresponding terminology to fix the ideas (see Figure
\ref{fig:Uplink-RoC-based}).

\subsection{RAN Model\label{subsec:RAN-Model}}
We consider the same Wyner-type radio access model of \cite{Kassab}, which is
described in this subsection. The Wyner model is an abstraction of cellular systems that captures one of the main aspects of such settings, namely the locality of inter-cell interference. The advantage of employing such a simple model is the possibility to obtain analytical insights, which is a first mandatory step for the performance assessment under more realistic operating conditions \cite{simeone2012cooperative}. As illustrated in Figure \ref{fig:Time-frequency-resource-allocati}, 
the direct channel gain from the eMBB UE and the EN belonging to the
same cell is set to one, while the inter-cell eMBB channel gain is
equal to $\alpha\in[0,1]$. Furthermore, the URLLC UEs have a channel
gain equal to $\beta>0$. The URLLC user is assumed to be in the proximity
of the EN, and thus it does not interfere with the neighboring cells.
The eMBB user, instead, is assumed to be located at the edges of the
cell in order to consider worst-case performance guarantees. As a
result, each eMBB user interferes with both left and right neighboring
cells, following the standard Wyner model \cite{simeone2012cooperative}. All channel gains are assumed
to be constant over the considered radio resources shown Figure \ref{fig:Time-frequency-resource-allocati},
and known to all UEs and ENs.
\begin{figure}
\begin{centering}
\includegraphics[width=0.5\columnwidth]{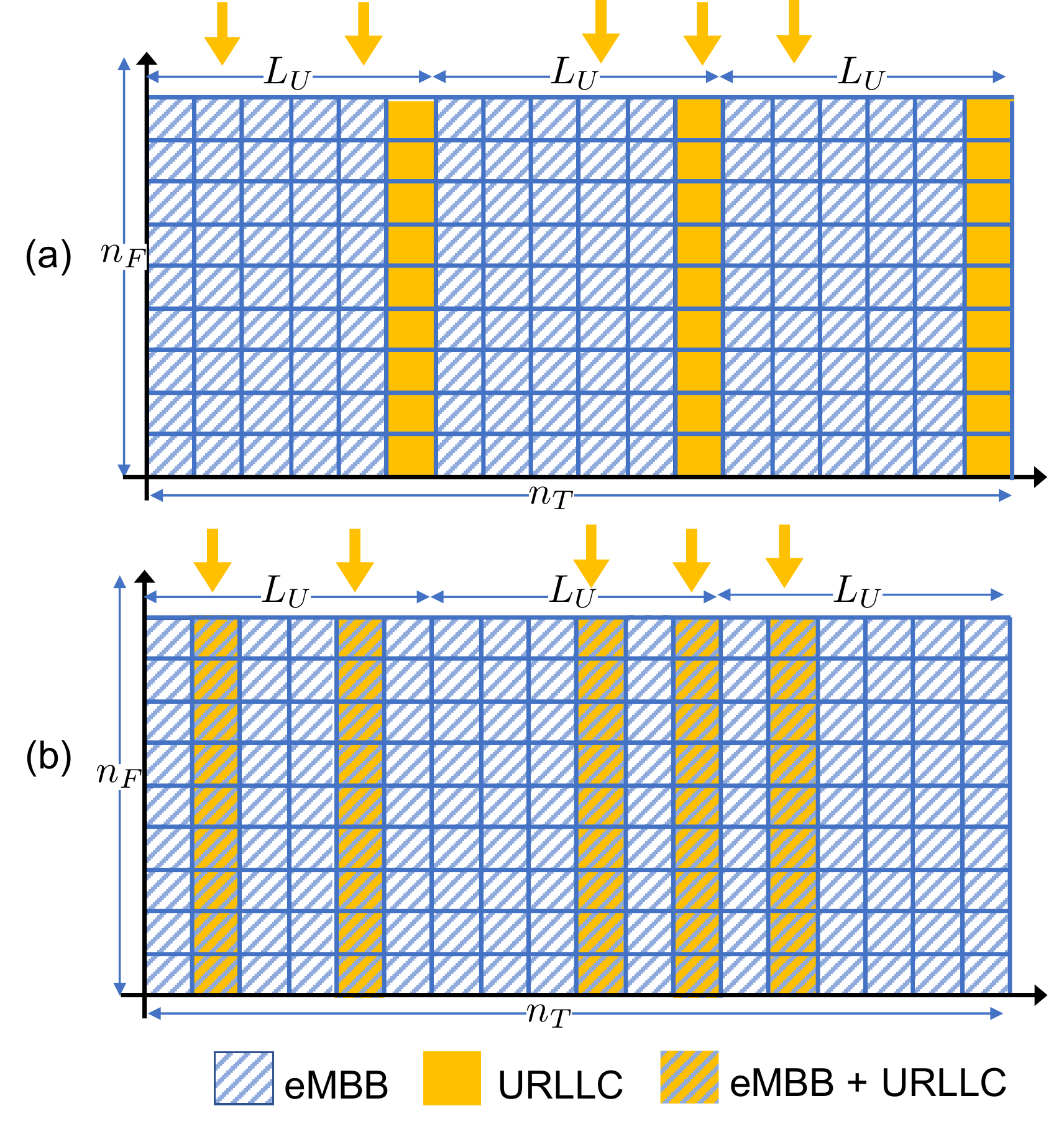}
\par\end{centering}
\caption{\label{fig:Time-frequency-resource-allocati}Time-frequency resource
allocation: ({\bf a}) Orthogonal Multiple Access (OMA) and ({\bf b}) Non-Orthogonal
Multiple Access (NOMA). Downwards arrows denote arrival of URLLC packets.}
\end{figure}

As illustrated in Figure \ref{fig:Time-frequency-resource-allocati},
we assume that the time-frequency plane is divided in $n_{T}$ minislots,
indexed as $t\in[1,n_{T}]$, where each minislot is composed of $n_{F}$
frequency channels, indexed as $f\in[1,n_{F}]$, for a total of $n_{F}n_{T}$
time-frequency radio resources. Each radio resource accommodates the
transmission of a single symbol, although generalizations are straightforward.
The eMBB UEs transmit over the entire $n_{F}\times n_{T}$ time-frequency
frame. In contrast, due to the latency constraints of the URLLC traffic,
each URLLC transmission is limited to the $n_{F}$ frequency channels
of a single minislot, and URLLC packets are generally small compared
to the eMBB frame, which requires the condition $n_{T}\gg1$. As illustrated
in Figure \ref{fig:Time-frequency-resource-allocati}, each URLLC UE
generates an independent packet in each minislot with probability
$q$. This packet is transmitted at the next available transmission
opportunity in a grant-free manner.

In the case of OMA, one minislot is exclusively
allocated for URLLC transmission every $L_{U}$ minislots. Parameter
$L_{U}$ is considered here as the worst-case access latency. Accordingly,
if more than one packet is generated within the $L_{U}$ minislots
between two transmission opportunities, only one of those packets is
randomly selected for transmission and all the remaining are discarded.
The signal $Y_{k}^{f}(t)$ received at the $k$-th EN at the $f$-th
frequency under OMA is

\begin{equation}
Y_{k}^{f}(t)=\begin{cases}
\beta A_{k}(t)U_{k}^{f}(t)+Z_{k}^{f}(t), & \text{if}\,t=L_{U},~2L_{U},~...\\
X_{k}^{f}(t)+\alpha X_{k-1}^{f}(t)+\alpha X_{k+1}^{f}(t)+Z_{k}^{f}(t), & \text{otherwise}
\end{cases}\label{eq:OMA received signal}
\end{equation}
where $X_{k}^{f}(t)$ and $U_{k}^{f}(t)$ are the signals transmitted
at time $t$ and subcarrier $f$ by the $k$-th eMBB UE and URLLC
UE, respectively; $Z_{k}^{f}(t)\sim\mathcal{CN}(0,1)$ is the unit-power
zero-mean additive white Gaussian noise; and $A_{k}(t)\in\{0,1\}$
is a binary variable indicating whether or not the URLLC UE is transmitting
at time $t$. 

In case of NOMA, the URLLC UE transmits
its packet in the same slot where it is generated by the application
layer, so that the access latency is always minimal, i.e., $L_{U}=1$
minislot. Under NOMA, the signal $Y_{k}^{f}(t)$ received at the $k$-th
EN at the $f$-th frequency is
\begin{equation}
Y_{k}^{f}(t)=X_{k}^{f}(t)+\alpha X_{k-1}^{f}(t)+\alpha X_{k+1}^{f}(t)+\beta A_{k}(t)U_{k}^{f}(t)+Z_{k}^{f}(t).\label{eq:NOMA received signal}
\end{equation}
According to the circulant Wyner model, in (\ref{eq:OMA received signal}) and (\ref{eq:NOMA received signal}), we assume that $[k-1]=M$ for $k=1$ and $[k+1]=1$ for $k=M$, in order to guarantee symmetry.

For both OMA and NOMA, the power constraints for the $k$-th eMBB
and URLLC users are defined within each radio resource frame respectively
as 
\begin{equation}
\frac{1}{n_{F}n_{T}}\sum_{t=1}^{n_{T}}\sum_{f=1}^{n_{F}}\mathbb{E}\left[\left|X_{k}^{f}(t)\right|^{2}\right]\leq P_{B},
\end{equation}
and
\begin{equation}
\frac{1}{n_{F}}\sum_{f=1}^{n_{F}}\mathbb{E}\left[\left|U_{k}^{f}(t)\right|^{2}\right]\leq P_{U},
\end{equation}
where the temporal average are taken over all symbols within a codeword.

Models (\ref{eq:OMA received signal}) and (\ref{eq:NOMA received signal})
can be written in matrix form as
\begin{equation}
\mathbf{Y}(t)=\mathbf{X}(t)\mathbf{H}+\beta\mathbf{U}(t)\mathbf{A}(t)+\mathbf{Z}(t),\label{eq:-14}
\end{equation}
{where matrix $\mathbf{Y}(t)=[\mathbf{y}_{1}(t),~\mathbf{y}_{2}(t),~...,~\mathbf{y}_{M}(t)]\in\mathbb{C}^{n_{F}\times M}$
~gathers all the signals received at all the $M$ ENs over all the
$n_{F}$ frequencies, and the $k$-th column $\mathbf{y}_{k}(t)\in\mathbb{C}^{n_{F}\times1}$
denotes the signal received across all the radio frequencies at the
$k$-th EN. The channel matrix $\mathbf{H}\in\mathbb{R}^{M\times M}$
is circulant with the first column given by vector $[1,~\alpha,~0,~...,~0,~\alpha]^{T}$;
matrices $\mathbf{U}(t)\in\mathbb{C}^{n_{F}\times M}$ and $\mathbf{X}(t)\in\mathbb{C}^{n_{F}\times M}$
collect the signals transmitted by URLLC and eMBB UEs, respectively;
and $\mathbf{Z}(t)\in\mathbb{C}^{n_{F}\times M}$ is the overall noise
matrix. Finally, $\mathbf{A}(t)$ is a diagonal matrix whose $k$-th
diagonal element is a Bernoulli random variable distributed as $A_{k}(t)\sim\mathcal{B}(q),\,\forall k=1,~2,~...,~M$. 

\subsection{Space-Frequency Analog Fronthaul Channel}

In the considered analog fronthaul architecture, the $k$-th EN forwards
the signal $\mathbf{y}_{k}(t)$ received by the UEs to the BBU over
a wired-access link in a fully analog fashion. As depicted in Figure
\ref{fig:Uplink-RoC-based}, we focus our attention on a multichannel
link that is possibly affected by inter-channel interference. While
the model considered here can apply also to wireless multichannel
links, as in Figure \ref{fig:C-RAN-architecture-overview:}, we adopt
here the terminology of Analog Radio-over-Copper (A-RoC) as an important
example in which A-RoX is affected by fronthauling inter-link interference.
Accordingly, each of the cables employed for the fronthaul contains
$l_{S}$ twisted-pairs, i.e., $l_{S}$ space-separated channels, indexed
as $c\in[1,l_{S}]$. Each pair carries a bandwidth equal to $l_{F}\leq n_{F}$
frequency channels of the RAN, indexed as $f'\in[1,l_{F}]$, so that
a total of $l_{S}l_{F}$ space-frequency resource blocks are available
over each cable, as shown in Figure \ref{fig:Space-frequency-cable-resource}.
Furthermore, we assume that each analog fronthaul link has enough
resources to accommodate the transmission of the whole radio signal
at each EN, i.e., $l_{S}l_{F}\geq n_{F}.$
\begin{figure}
\begin{centering}
\includegraphics[width=0.4\columnwidth]{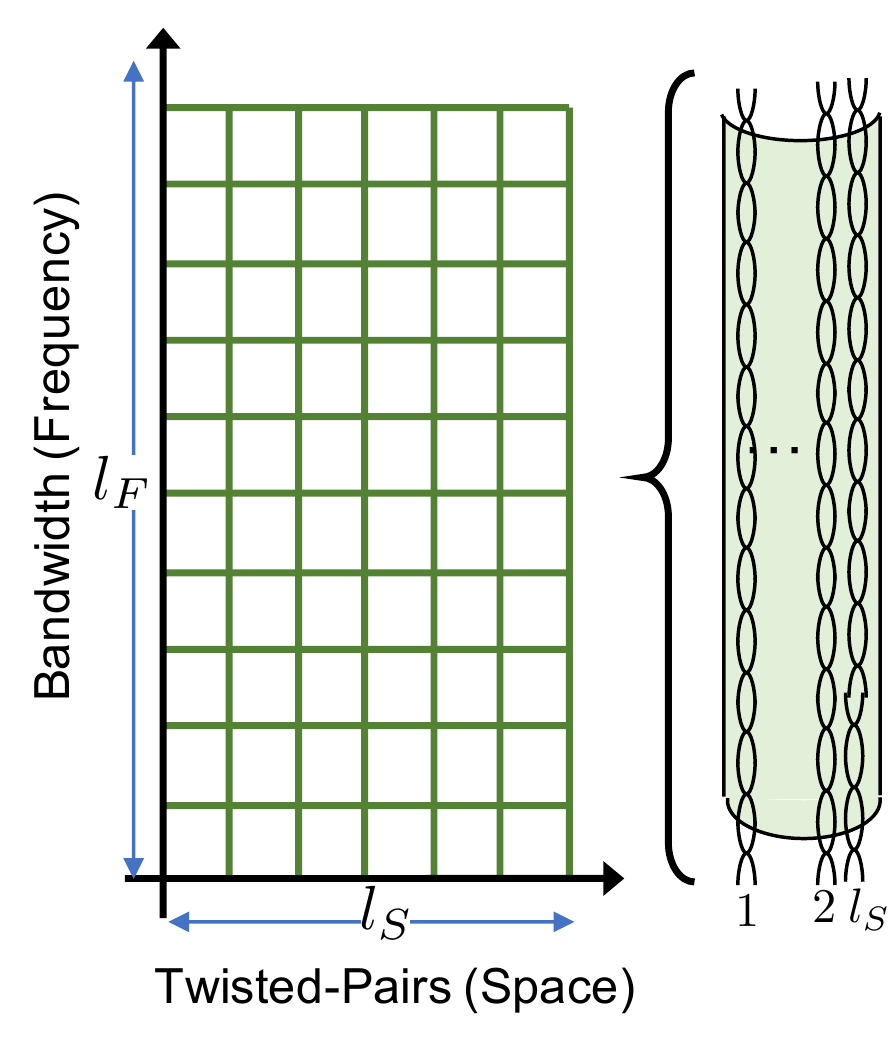}
\par\end{centering}
\caption{\label{fig:Space-frequency-cable-resource}Space-frequency cable resource
allocation}
\end{figure}
The fronthaul channel between each EN and the BBU is described by
the matrix $\mathbf{H}_{c}\in\mathbb{R}^{l_{S}\times l_{S}}$, which
accounts for direct channel gains on each cable, given by the diagonal
elements $[\mathbf{H}_{c}]_{ii}$, and for the intra-cable crosstalk,
described by the off-diagonal elements $[\mathbf{H}_{c}]_{ij}$, with
$i\neq j$. We assume here that the channel coefficients in $\mathbf{H}_{c}$
do not depend on frequency $f'$. Furthermore, in keeping the modeling
assumptions of the Wyner model, we posit that the direct channel gains
for all the pairs are normalized to 1, while the crosstalk coefficients
between any pair of twisted-pairs are given by a coupling parameter
$\gamma\geq0$. It follows that the fronthaul channel matrix can be
written as
\begin{equation}
\mathbf{H}_{c}=\gamma\mathbf{1}_{l_{S}}\mathbf{1}_{l_{S}}^{T}+(1-\gamma)\mathbf{I}_{l_{S}},\label{eq:-21}
\end{equation}
where $\mathbf{1}_{n}$ denotes a column vector of size $n$ of all
ones and $\mathbf{I}_{n}$ is the identity matrix of size $n$. We
note that, in case of wireless fronthaul links such as in A-RoR, the
coefficient $\gamma$ accounts for the mutual interference between
spatially separated radio links. As a result, one typically has $\gamma>0$
when considering sub-6 GHz frequency bands, while the condition $\gamma=0$
may be reasonable in the mmWave or THz bands, in which communication
is mainly noise-limited due to the highly directive beams \cite{Rappaport-etal_2013}.

For a given time $t$, the symbols $\mathbf{y}_{k}(t)$ received at
EN $k$-th over all the $n_{F}$ radio frequency channels are transported
to the BBU over the $l_{S}l_{F}$ cable resource blocks, where the
mapping between radio and cable resources is to be designed (see Section
\ref{subsec:Radio-Resource-Mapping}) and depends on the bandwidth
$l_{F}$ available at each twisted-pair. 

In this regards, we define the fraction $\mu\in[1/l_{S},1]$ of the
radio bandwidth $n_{F}$ that can be carried by each pair, referred
to as \textit{normalized cable bandwidth}, as
\begin{equation}
\mu=\frac{l_{F}}{n_{F}}.
\end{equation}
As a result, the quantity 
\begin{equation}
\eta=\mu\cdot l_{S}\label{eq:-37}
\end{equation}
 expresses the bandwidth amplification factor (or redundancy) over
the cable fronthaul, as $\eta\geq1$. To simplify, we assume here
that both $1/\mu$ and $\eta$ are integer numbers. The two extreme
situations with $\mu=1$, or $\eta=l_{S}$, and $\mu=1/l_{S}$, or
$\eta=1$, are shown in Figure \ref{fig:Mapping-of-radio} for $l_{S}=4$
twisted-pairs and $n_{F}=8$ subcarriers. For the first case, one
replica of the whole radio signal $\mathbf{y}_{k}(t)$ can be transmitted
over all of the $l_{S}$ pairs, and the bandwidth amplification over
cable is $\eta=l_{S}$. In the second case, disjoint fractions of
the received bandwidth can be forwarded on each pair and the bandwidth
amplification factor is $\eta=1$.
\begin{figure}
\begin{centering}
\includegraphics[width=0.7\textwidth]{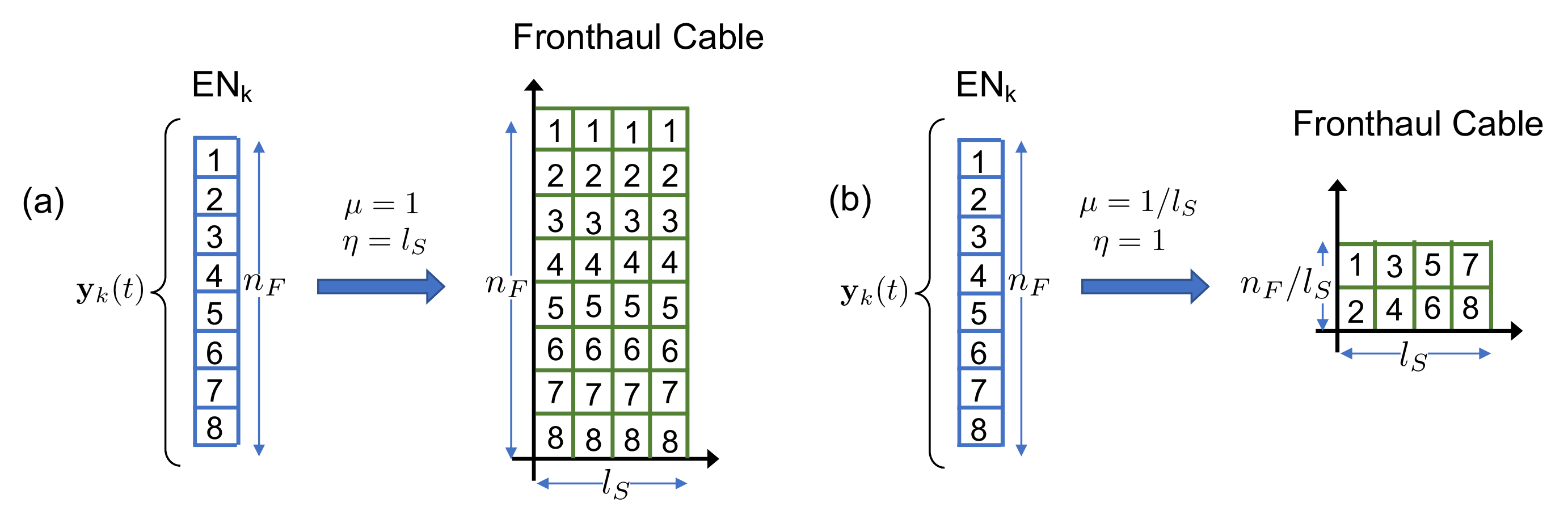}
\par\end{centering}
\caption{\label{fig:Mapping-of-radio} Mapping of radio resources over cable
resources: ({\bf a}) maximum normalized cable bandwidth (full redundancy), $\mu=1$, or $\eta=l_S$; ({\bf b}) minimal normalized cable bandwidth (no redundancy), $\mu=1/l_S$, or $\eta=1$.}

\end{figure}
We now detail the signal model for fronthaul transmission. To this
end, let us define the $l_{F}\times l_{S}$ matrix $\mathbf{\tilde{Y}}_{k}$
containing the signal to be transmitted by the EN $k$-th to the BBU
over the copper cable as 
\begin{equation}
\mathbf{\tilde{Y}}_{k}=[\tilde{\mathbf{y}}_{k}^{1}(t),\tilde{\mathbf{y}}_{k}^{2}(t),...,\tilde{\mathbf{y}}_{k}^{l_{S}}(t)],\label{eq:-4}
\end{equation}
where the $k$-th column $\tilde{\mathbf{y}}_{k}^{c}(t)\in\mathbb{C}^{l_{F}\times1}$
denotes the signal transmitted on twisted-pair $c$ across all the
$l_{F}$ cable frequency resources. The signal $\mathbf{\tilde{R}}_{k}\in\mathbb{C}^{l_{F}\times l_{S}}$
received at the BBU from the $k$-th EN across all the cable space-frequency
resources is then computed as
\begin{equation}
\tilde{\mathbf{R}}_{k}(t)=\tilde{\mathbf{Y}}_{k}(t)\mathbf{H}_{c}+\tilde{\mathbf{W}}_{k}(t),\label{eq:-1}
\end{equation}
where $\tilde{\mathbf{W}}_{k}(t)=[\tilde{\mathbf{w}}_{k}^{1}(t),\tilde{\mathbf{w}}_{k}^{2}(t),...,\tilde{\mathbf{w}}_{k}^{l_{S}}(t)]\in\mathbb{C}^{l_{F}\times l_{S}}$
is the additive white Gaussian cable noise uncorrelated over cable
pairs and frequencies, i.e., $\tilde{\mathbf{w}}_{k}^{c}(t)\sim\mathcal{CN}(\mathbf{0},\mathbf{I}_{l_{F}})$
for all pairs $c=1,~2,~...,~l_{S}$. 

As commonly assumed in wireline communications to control cable radiations
\cite{Hekrdla-Matera-etal_2015}, the power of the cable symbol $\left[\tilde{\mathbf{y}}_{k}^{c}(t)\right]_{f^{'}}$
transmitted from EN$_{k}$ over twisted-pair $c$ at frequency $f^{'}$
is constrained to $P_{c}$ by the (short-term: One can also consider the ``long-term'' power constraint $n_{T}^{-1}\sum_{t}\mathbb{E}[|[\tilde{\mathbf{y}}_{k}^{c}(t)]_{f^{'}}|^{2}]\leq~P_{c},\,\forall c\in[1,l_{S}],\thinspace f'\in[1,l_{F}]$
with minor modifications to the analysis and final results. ) power constraint 
\begin{equation}
\mathbb{E}\left[\left|\left[\tilde{\mathbf{y}}_{k}^{c}(t)\right]_{f^{'}}\right|^{2}\right]\leq P_{c}\quad\forall c\in[1,l_{S}],f'\in[1,l_{F}],t\in[1,n_{T}].\label{eq:-3}
\end{equation}

In the following, we will omit the time index $t$, when no confusion
arises.

\subsection{Performance Metrics}

The performance metrics used to evaluate the interaction between eMBB
and URLLC services in the considered A-RoC-based C-RAN architecture
are detailed in the following. 

\subsubsection{eMBB }

Capacity enhancement is the main goal of the eMBB service, which is
envisioned to provide very high-rate communication to all the UEs.
Therefore, for eMBB UEs, we are interested in the per-UE rate defined
as
\begin{equation}
R_{B}=\frac{\text{log}_{2}(M_{B})}{n_{T}n_{F}},\label{eq:-31}
\end{equation}
where $M_{B}$ is the number of codewords in the codebook of each
eMBB UE. 

\subsubsection{URLLC}

Differently from eMBB, URLLC service is mainly focused on low-latency
and reliability aspects. Due to the short length of URLLC packets,
in order to guarantee ultra reliable communications, we need to ensure
that the error probability for each URLLC UE, denoted as $\text{Pr}[E_{U}]$,
is bounded by a predefined value $\epsilon_{U}$ (typically smaller
than $10^{-3}$) as 
\begin{equation}
\text{Pr}[E_{U}]\leq\epsilon_{U}.\label{eq:-33}
\end{equation}
Concerning latency, we define the maximum access latency $L_{U}$
as the maximum number of minislots that an URLLC UE has to wait before
transmitting a packet. Finally, although rate enhancement is not one
of the goals of URLLC service, it is still important to evaluate the
per-UE URLLC rate that can be guaranteed while satisfying the aforementioned
latency and reliability constraints. Similarly to (\ref{eq:-31}),
the per-UE URLLC rate is defined as 
\begin{equation}
R_{U}=\frac{\text{log}_{2}(M_{U})}{n_{F}},
\end{equation}
where $M_{U}$ is the number of URLLC codewords in the codebook used
by the URLLC UE for each information packet.

\section{Analog Fronthaul Signal Processing\label{sec:Cable-Fronthaul-Signal}}

The analog fronthaul links employed in the C-RAN system under study
pose several challenges in the system design, which are addressed
in this section. Firstly, the radio signal received at each EN needs
to be mapped over the corresponding fronthaul resources in both frequency
and space dimension. Secondly, the signal at the output of each fronthaul
link needs to be processed in order to maximize the Signal-to-Noise
Ratio (SNR) for all UE signals. Finally, the power constraints in
(\ref{eq:-3}) must be properly enforced. All these requirements are
to be addressed by all-analog processing in order to meet the low-complexity
and latency constraints of the analog fronthaul. In the rest of this
section, we discuss each of these problems in turn.

\subsection{\label{subsec:Radio-Resource-Mapping}Radio Resource Mapping over
Fronthaul Channels}

To maximize the SNRs for all the signals forwarded over the
fronthaul by symmetry, we need to ensure that: \textit{(i) }all the
received signals are replicated $\eta$ times across the cable twisted-pairs,
where we recall that $\eta$ is the bandwidth amplification factor
defined in (\ref{eq:-37}); and \textit{(ii)} cable cross-talk interference
among different radio frequency bands is minimized. In fact, as the
transmitted power at the cable input is limited by the constraints
in (\ref{eq:-3}), a simple and effective way to cope with the impairments of the analog fronthaul 
links using the only analog-processing capability is by introducing redundancy in the fronthaul transmission. 
To this end, without loss of generality, we assume the following mapping rule between the $n_{F}$ radio signals
at each EN and the $l_{S}l_{F}$ cable resources. 

Let us consider the radio signal $\mathbf{y}_{k}=\left[Y_{k}^{1},~Y_{k}^{2},~...,~Y_{k}^{n_{F}}\right]^{T}$
received at the $k$-th EN. For a given normalized cable bandwidth
$\mu$, the $n_{F}$ frequency channels of the radio signal $\mathbf{y}_{k}$
can be split into $1/\mu$ sub-vectors of size $l_{F}=\mu n_{F}$
as
\begin{equation}
\mathbf{y}_{k}=\begin{bmatrix}\mathbf{y}_{k}^{1}\\
\mathbf{y}_{k}^{2}\\
\vdots\\
\mathbf{y}_{k}^{\frac{1}{\mu}}
\end{bmatrix},\label{eq:-5-1}
\end{equation}
where each vector $\mathbf{y}_{k}^{j}$ contains a disjoint fraction
of the radio signal bandwidth $n_{F}$. Each vector $\mathbf{y}_{k}^{j}$
can be transmitted over $\eta$ twisted-pairs, where we recall that
$\eta$ is the fronthaul redundancy factor. To this end, each signal
$\mathbf{y}_{k}^{j}$ in (\ref{eq:}) is transmitted over $\eta$
consecutive cable twisted-pairs.

To formalize the described mapping, the first step consists in reorganizing
the signal $\mathbf{y}_{k}$ into a $l_{F}\times\frac{1}{\mu}$ matrix
as 
\begin{equation}
\mathbf{Y}_{k}=\text{vec}_{\frac{1}{\mu}}^{-1}(\mathbf{y}_{k})=\begin{bmatrix}\mathbf{y}_{k}^{1} & \mathbf{y}_{k}^{2} & ... & \mathbf{y}_{k}^{\frac{1}{\mu}}\end{bmatrix},\label{eq:}
\end{equation}
where the operator $\text{vec}_{\frac{1}{\mu}}^{-1}(\mathbf{\cdot}):\,\mathbb{C}^{n_{F}}\rightarrow\mathbb{C}^{n_{F}\cdot\mu\times\frac{1}{\mu}}$
acts as the inverse of the vectorization operator $\text{vec}(\cdot)$,
with the subindex $1/\mu$ denoting the number of columns of the resulting
matrix. Then, the overall cable signal $\tilde{\mathbf{Y}}_{k}$ transmitted
by EN $k$-th in (\ref{eq:-4}) can be equivalently written as
\begin{equation}
\tilde{\mathbf{Y}}_{k}=\begin{bmatrix}\underset{\eta}{\underbrace{\begin{array}{cccc}
\mathbf{y}_{k}^{1} & \mathbf{y}_{k}^{1} & \dots & \mathbf{y}_{k}^{1}\end{array}}}, & \underset{\eta}{\underbrace{\begin{array}{cccc}
\mathbf{y}_{k}^{2} & \mathbf{y}_{k}^{2} & \dots & \mathbf{y}_{k}^{2}\end{array}}} & ... & \underset{\eta}{\underbrace{\begin{array}{cccc}
\mathbf{y}_{k}^{\frac{1}{\mu}} & \mathbf{y}_{k}^{\frac{1}{\mu}} & \dots & \mathbf{y}_{k}^{\frac{1}{\mu}}\end{array}}}\end{bmatrix},\label{eq:-2}
\end{equation}
or, in compact form, as 
\begin{equation}
\tilde{\mathbf{Y}}_{k}=\mathbf{Y}_{k}\otimes\mathbf{1}_{\eta}^{T},\label{eq:-5}
\end{equation}
where $\otimes$ denotes the Kronecker product (for a review of Kronecker
product properties in signal processing we refer the reader to \cite{spagnolini2018statistical}).
Notice that in case of full normalized cable bandwidth, i.e., $\mu=1$
(corresponding to $\eta=l_{S}$), the signal $\tilde{\mathbf{Y}}_{k}$
transmitted over the fronthaul cable simplifies to $\tilde{\mathbf{Y}}_{k}=\mathbf{y}_{k}\otimes\mathbf{1}_{l_{S}}^{T}$,
which implies that the radio signal $\mathbf{y}_{k}$ is replicated
over all the $l_{S}$ twisted-pairs. On the contrary, when the normalized
cable bandwidth is minimal, i.e., $\mu=1/l_{S}$ (corresponding to
$\eta=1$), the signal $\tilde{\mathbf{Y}}_{k}$ does not contain
any redundancy, and disjoints signals are transmitted over all pairs,
so that the cable signal $\tilde{\mathbf{Y}}_{k}$ equals the matrix
radio signal in (\ref{eq:}) as $\tilde{\mathbf{Y}}_{k}=\mathbf{Y}_{k}$. 

\begin{remark}
The easiest practical implementation of the proposed analog radio resource mapping at the EN is by grouping the subcarriers onto a specific frequency portion of the cable, as described in \cite{matera2017optimal,Matera-Combi-etal_2017,Matera-Spagnolini_2018,Naqvi-Matera-etal_2017}. As an example, let us assume that the EN is equipped with 5 antennas, that each antenna receives a 20-MHz radio signal, and that the analog fronthauling disposes of 4 links with 100 MHz bandwidth each. In this case, the above references have shown that it is possible to freely map, or to replicate, in an all-analog fashion the 
5 $\times$ 20 MHz
~bands onto the overall 4 $\times$ 100 MHz = 400 MHz fronthaul bandwidth. This example corresponds to a special case of the model studied in this paper, obtained by setting $\mu=1$, i.e., the whole radio signal bandwidth received at the ENs is mapped/replicated over the analog fronthauling. More generally, this paper posits the possibility to carry out the fronthaul mapping at a finer granularity, i.e., at a subcarrier level. In this case, filtering operations would in practice be mandatory in order to extract groups of subcarriers. This operation can be implemented in principle still by analog filters, whose design is left as future works.
\end{remark}

\subsection{Signal Combining at the Fronthaul Output }

As discussed, depending on the fronthaul bandwidth, a number $\eta$
of noisy replicas of the radio signals received at each EN are relayed
to the BBU over $\eta$ different twisted-pairs. Hence, in order to
maximize the SNRs for all signals, Maximum Ratio Combining (MRC) \cite{jakes1994microwave}
is applied at the cable output as
\begin{equation}
\mathbf{R}_{k}=\tilde{\mathbf{R}}_{k}\mathbf{G},\label{eq:-27}
\end{equation}
where $\mathbf{R}_{k}\in\mathbb{C}^{l_{F}\times\frac{1}{\mu}}$ is
the signal received at BBU from the $k$-th EN after the combiner
and $\mathbf{G}\in\mathbb{R}^{l_{S}\times\frac{1}{\mu}}$ is the MRC
matrix. Under the assumptions here, MRC coincides with equal ratio
combining and hence matrix $\mathbf{G}$ can be written as 
\begin{equation}
\mathbf{G}=\frac{1}{\eta}\left(\mathbf{I}_{\frac{1}{\mu}}\otimes\mathbf{1}_{\eta}\right).\label{eq:-20}
\end{equation}

As an example, in the case of maximum redundancy, i.e., $\eta=l_{S}$,
the MRC matrix $\mathbf{G}=l_{S}^{-1}\mathbf{1}_{l_{S}}$ combines the analog signals received over all pairs, since they
carry the same information signal. On the contrary, in the case of
minimal normalized bandwidth $\mu=1/l_{S}$ or $\eta=1$, the matrix
$\mathbf{G}$ equals the identity matrix as $\mathbf{G}=\mathbf{I}_{l_{S}}$,
since no combining is possible.

The signal $\mathbf{r}_{k}\in\mathbb{C}^{n_{F}}$ received at the
BBU from EN$_{k}$ across the $n_{F}$ subcarriers is thus obtained
by vectorizing matrix $\mathbf{R}_{k}$ in (\ref{eq:-27}) as 

\begin{equation}
\mathbf{r}_{k}=\text{vec}(\mathbf{R}_{k}).\label{eq:-38}
\end{equation}

The relationship between the signal $\mathbf{r}_{k}$ (\ref{eq:-38})
obtained at the output of the combiner and the radio received signal
$\mathbf{y}_{k}$ in (\ref{eq:-5-1}) is summarized by the block-scheme
in Figure \ref{tab:Tranmsission-over-cable}, and it is
\begin{equation}
\mathbf{r}_{k}=\text{vec}\left[\left(\left(\text{vec}_{\frac{1}{\mu}}^{-1}\left(\mathbf{y}_{k}\right)\otimes\mathbf{1}_{\eta}^{T}\right)\mathbf{H}_{c}+\tilde{\mathbf{W}}_{k}\right)\mathbf{G}_{k}\right].
\end{equation}

\begin{figure}
\begin{centering}
\includegraphics[width=0.9\textwidth]{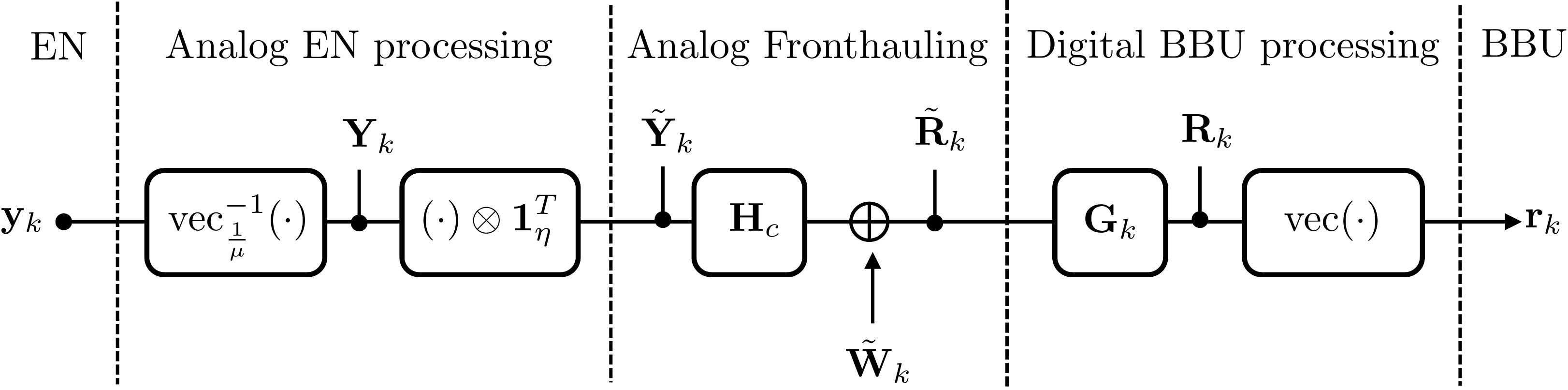}
\par\end{centering}
\caption{\label{tab:Tranmsission-over-cable}Relationship between the signal
$\mathbf{r}_{k}$ (\ref{eq:-38}) obtained at the output of the combiner
and the radio received signal $\mathbf{y}_{k}$ in (\ref{eq:-5-1}).}
\end{figure}
Finally, we collect the overall signal $\mathbf{R}\in\mathbb{C}^{n_{F}\times M}$
received at the BBU from all ENs across all frequencies in the matrix
\begin{equation}
\mathbf{R}=[\mathbf{r}_{1},\mathbf{r}_{2},~...,~\mathbf{r}_{M}].\label{eq:-15}
\end{equation}
After some algebraic manipulations, it is possible to express Equation
(\ref{eq:-15}) in a more compact form, which is reported in the following
Lemma \ref{lem:In-C-RAN-architecture}. 
\begin{lem}
\label{lem:In-C-RAN-architecture}In the given C-RAN architecture
with analog fronthaul links, for a given bandwidth amplification factor
$\eta\geq1$, the signal $\mathbf{R}\in\mathbb{C}^{n_{F}\times M}$
received at the BBU from all ENs across all radio frequencies after
MRC can be written as
\begin{equation}
\mathbf{R}=\left(\mathbf{H}_{c}^{\eta}\otimes\mathbf{I}_{l_{F}}\right)\mathbf{Y}+\mathbf{W},\label{eq:-28}
\end{equation}
where 
\begin{equation}
\mathbf{H}_{c}^{\eta}=\gamma\eta\mathbf{1}_{\frac{1}{\mu}}\mathbf{1}_{\frac{1}{\mu}}^{T}+(1-\gamma)\mathbf{I}_{\frac{1}{\mu}}\label{eq:-6}
\end{equation}
 is the equivalent fronthaul channel matrix; $\mathbf{Y}$ is the
signal received at all ENs across all frequency radio channels in
(\ref{eq:-14}); and $\text{\ensuremath{\mathbf{W}}}=[\mathbf{w}_{1},\mathbf{w}_{2},...,\mathbf{w}_{M}]$
is the equivalent cable noise at the BBU after MRC, with the $k$-th
column distributed as $\mathbf{w}_{k}\sim\mathcal{CN}\left(\mathbf{0},\frac{1}{\eta}\mathbf{I}_{n_{F}}\right)$
for all $k=1,2,...,M$.
\end{lem}
\begin{proof}
see Appendix \ref{subsec:Example:-Signal-Combining}.
\end{proof}
To gain some insights, it is useful again to consider the
two extreme cases of maximum redundancy, i.e., $\eta=l_{S}$, and no
redundancy, i.e., $\eta=1$. In the former case, the equivalent channel
(\ref{eq:-6}) equals the scalar $\mathbf{H}_{c}^{\eta}=1+\gamma(l_{S}-1)$.
This demonstrates the effect of transmitting a replica of the whole
radio signal over all pairs. In fact, the useful signal is received
at the BBU not only through the direct path, which has unit gain,
but also from the remaining $l_{S}-1$ interfering paths, each with
gain $\gamma$, which constructively contribute to the overall SNR
after the combiner. More precisely, it can be observed that, in case
of full redundancy, the SNR of the radio signal $\mathbf{Y}$ at the
BBU is increased by the analog fronthaul links by a factor of $(1+\gamma(l_{S}-1))^{2}/(1/\eta)=l_{S}(1+\gamma(l_{S}-1))^{2}$.
As a result, in this case, for a coupling factor $\gamma>0$, the
SNR at the BBU increases with the cube of the number of fronthaul
links $l_{S}$. In contrast, for $\mu=1/l_{S}$, the equivalent fronthaul
channel reflects the fact that signals forwarded over the different
pairs interfere with each other, and is equal to $\mathbf{H}_{c}^{\eta}=\mathbf{H}_{c}$.
The beneficial effect of redundantly transmitting radio signals over
different pairs is reflected also in the power of the noise after
the combiner, which is reduced proportionally to the bandwidth amplification
factor $\eta$. 

\subsection{Fronthaul Power Constraints}

To enforce the cable power constraints in (\ref{eq:-3}),
it is necessary to scale the radio signal $\mathbf{Y}$ in (\ref{eq:-28})
by a factor of $\lambda$ prior to the transmission over the fronthaul.
This is given as
\begin{equation}
\lambda=\sqrt{\frac{P_{c}}{\delta P_{B}(1+2\alpha^{2})+1}},\label{eq:-18-2}
\end{equation}
where $\delta$ is equal to $\delta=\left(1-L_{U}^{-1}\right)^{-1}$ for OMA, accounting
for the fact that only $L_{U}-1$ minislots are devoted to the eMBB
UE, while it equals $\delta=1$ for NOMA, since the eMBB transmission
spreads over all $L_{U}$ minislots. To simplify the notation, in the
following we will account for the gain $\lambda$ by scaling the noise
over the cable after MRC in Equation (\ref{eq:-6}) accordingly as 
\begin{equation}
\mathbf{w}_{k}(t)\sim\mathcal{CN}\left(\mathbf{0},\frac{1}{\eta\lambda^{2}}\mathbf{I}\right).\label{eq:-18}
\end{equation}

\section{Orthogonal Multiple Access (OMA)\label{sec:Orthogonal-Multiple-Access}}

As described in Section \ref{subsec:RAN-Model}, under OMA over the radio
channel, one minislot every $L_{U}$ is exclusively allocated to URLLC
UEs, while eMBB UEs transmit over the remaining minislots. In this
way, URLLC UEs never interfere with eMBB transmissions. If more than
one URLLC packet is generated at a user between two URLLC transmission
opportunities, only one of such packets (randomly selected) is transmitted,
while the others are discarded, causing a blockage error. Due to the
latency constraints, URLLC signals are digitized and decoded locally
at the ENs, while the eMB signals are first mapped over the fronthaul
lines, and then analogically forwarded to the BBU, as mathematically
summarized in Figure \ref{fig:Orthogonal-Multiple-Access}. In this
section, we derive the expressions for the eMBB and URLLC rates under
OMA for a given URLLC access latency $L_{U}$ and, in the case of
URLLC, for a fixed URLLC target error probability $\epsilon_{U}$.
\begin{figure}
\begin{centering}
\includegraphics[width=0.5\paperwidth]{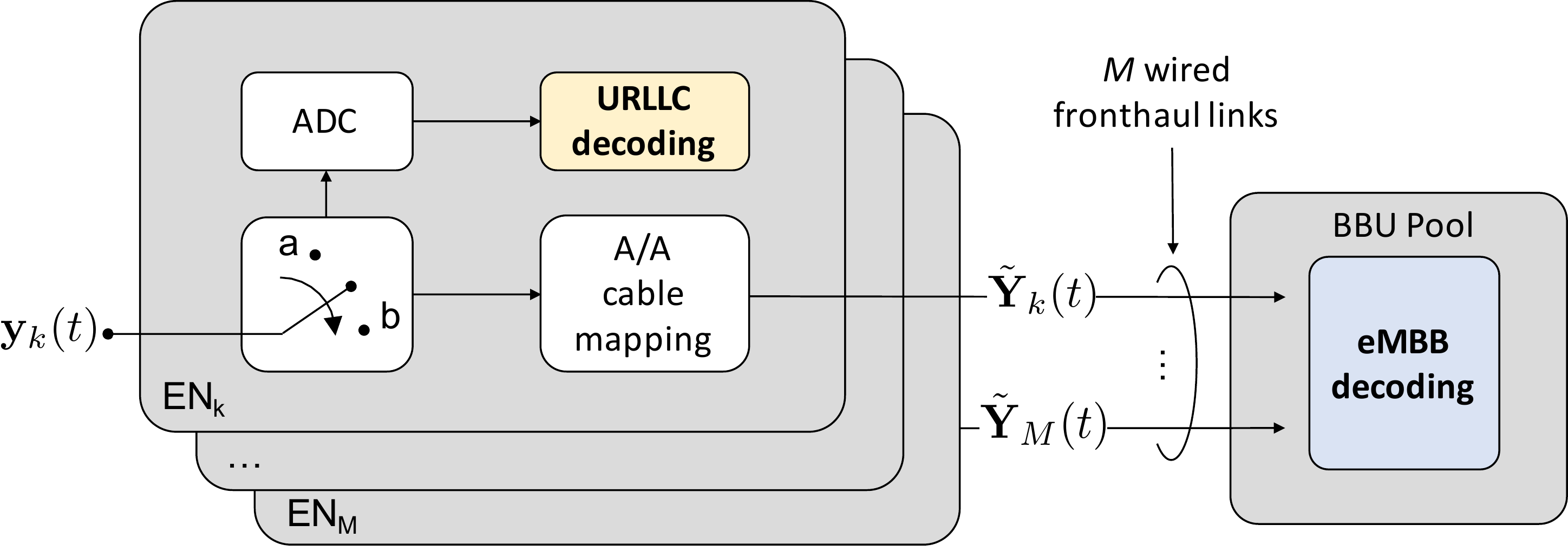}
\par\end{centering}
\caption{\label{fig:Orthogonal-Multiple-Access}Block diagram of the operation
of the ENs and BBU for Orthogonal Multiple Access (OMA). A/A stands for Analog-to-Analog.}
\end{figure}

\subsection{URLLC Rate}

To evaluate the per-UE URLLC rate under OMA and for a given
URLLC target error probability $\epsilon_{U}$, we follow the approach
in \cite{Kassab}, which is reviewed here. URLLC packets are generally
short due to the strict latency constraints, and the maximum achievable
rate can be computed by leveraging results from finite blocklength
information theory. To this end, fix a given blocklength $n_{F}$
and URLLC error decoding probability $\epsilon_{U}^{D}$. Notice that
this probability is different from the general URLLC target error
probability $\epsilon_{U}$, as detailed later in this section. According
to \cite{Polyanskiy-Poor-Verdu_2010}, the URLLC rate can be well
approximated by
\begin{equation}
R_{U}=\text{log}_{2}(1+\beta^{2}P_{U})-\sqrt{\frac{V}{n_{F}}}Q^{-1}(\epsilon_{U}^{D}),\label{eq:-34}
\end{equation}
where 
\begin{equation}
V=\frac{\beta^{2}P_{U}}{1+\beta^{2}P_{U}}
\end{equation}
is the channel dispersion and $\mathcal{Q}^{-1}(\cdot)$ is the inverse
Q-function (see Section \ref{subsec:Notation}).

The error probability for URLLC packets is the sum of two contributions.
The first represents the probability that an URLLC packet is discarded
due to blockage, given that only one URLLC packet can be transmitted
within the required $L_{U}$ worst-case latency; while the second
is the probability that the packet is transmitted but not successfully
decoded. Accordingly, the overall error probability can be computed
as
\begin{equation}
\text{Pr}[E_{U}]=\sum_{n}^{L_{U}-1}p(n)\frac{n}{n+1}+\sum_{n}^{L_{U}-1}p(n)\frac{1}{n+1}\epsilon_{U}^{D},\label{eq:-32}
\end{equation}
where $p(n)=\text{Pr}[N_{U}(L_{U})=n]$ is the distribution of the
binomial random variable $N_{U}(L_{U})\sim\text{Bin}(L_{U}-1,q)$
representing the number of additional packets generated by the URLLC
UE during the remaining minislots between two transmission opportunities.
The decoding error probability $\epsilon_{U}^{D}$ in (\ref{eq:-34})
can be obtained from the URLLC reliability constraint in (\ref{eq:-33}),
i.e., $\text{Pr}[E_{U}]=\epsilon_{U}$.

\subsection{eMBB Rate}

The eMBB signals received at the ENs are forwarded over the analog
fronthaul to the BBU, where centralized digital signal processing
and decoding are performed. In the case of OMA, the eMBB signal is
free from URLLC interference, hence signal $\mathbf{Y}$ received
by all ENs over all radio channels in (\ref{eq:-14}) can be written
as
\begin{equation}
\mathbf{Y}=\mathbf{X}\mathbf{H}+\mathbf{Z}.\label{eq:-15-1}
\end{equation}
By substituting (\ref{eq:-15-1}) in (\ref{eq:-28}), it is possible
to compute the expression for the eMBB per-UE rate under OMA, as shown
in Lemma \ref{lem:In-the-case}. Notice that, unlike the case of URLLC
packets, the $n_{F}n_{T}$ blocklength of eMBB packets allows for
the use of standard asymptotic Shannon theory in the computation of
eMBB information rate.
\begin{lem}
\label{lem:In-the-case}In the given C-RAN architecture with analog
fronthaul links, for a given bandwidth amplification factor $\eta\geq1$,
the eMBB user rate under OMA is given as
\begin{equation}
R_{B}=\mu\frac{1-L_{U}^{-1}}{M}\text{\text{log}}\left(\text{det}\left(\mathbf{I}+\bar{P}_{B}\mathbf{R}_{z_{\text{eq}}}^{-1}\mathbf{H}_{\text{eq}}\mathbf{H}_{\text{eq}}^{T}\right)\right),\label{eq:-36}
\end{equation}
where $\bar{P}_{B}=P_{B}\left(1-L_{U}^{-1}\right){}^{-1}$ is the
transmission power of eMBB users under OMA, $\mathbf{H}_{\text{eq}}=\mathbf{H}\otimes\mathbf{H}_{c}^{\eta}$
is the overall channel matrix comprising both the radio channel $\mathbf{H}$
and the equivalent cable channel $\mathbf{H}_{c}^{\eta}$ defined
in Lemma \ref{lem:In-C-RAN-architecture}, and $\mathbf{R}_{z_{\text{eq}}}=\mathbf{I}_{M}\otimes\mathbf{H}_{c}^{\eta}\mathbf{H}_{c}^{\eta}+\frac{1}{\lambda^{2}\eta}\mathbf{I}_{\frac{M}{\mu}}$
is the overall wireless plus cable noise at the BBU.
\end{lem}
\begin{proof}
see Appendix. \ref{Lemma2}.
\end{proof}
As a first observation, the eMBB rate (\ref{eq:-36}) linearly scales
with the normalized bandwidth $\mu$. This shows that a potential
performance degradation in terms of spectral efficiency can be incurred
in the presence of fronthaul channels with bandwidth limitations,
i.e., with $\mu<1$. This loss is pronounced in the presence of significant
inter-channel interference, i.e., for large $\gamma$. In fact, a
large $\gamma$ increases the effective noise power as per expression
of matrix $\mathbf{R}_{z_{\text{eq}}}$. It is also important to point
out that in the considered C-RAN system based on the analog relaying
of radio signals, the overall noise at the BBU is no longer white
as it accounts both for the white cable noise and the wireless noise,
where the latter is correlated when there is some bandwidth redundancy,
i.e., when $\mu\geq1/l_{S}$ or $\eta>1$.

\section{Non-Orthogonal Multiple Access (NOMA)\label{sec:Non-Orthogonal-Multiple-Access}}

In NOMA, URLLC UEs transmit in the same minislot where the packet
is generated, and hence the access latency is minimal and limited
to $L_{U}=1$ minislot. However, the URLLC signals mutually interfere
with the eMBB transmission, which spans the whole time-frequency resource
plane. Due to URLLC latency constraints, the eMBB signals necessarily
need to be treated as noise while decoding URLLC packets at the ENs.
On the contrary, several strategies can be adapted in order to deal with 
the interfering URLLC signal. Beside
puncturing, considered for 5G NR standardization
\cite{qualcomm_puncturing,3GPP_FinalReport_Puncturing}, this work
considers two other techniques, namely Treating Interference as Noise
(TIN) and Successive Interference Cancellation (SIC), as detailed
in the rest of this section.

\subsection{URLLC Rate under NOMA}

The URLLC per-UE rate for NOMA can be computed by leveraging results
from finite blocklength information theory similarly to the OMA case,
but accounting for the additional eMBB interference~\cite{scarlett2017dispersion}.
The URLLC per-UE rate under NOMA is thus well approximated by \cite{Kassab}
\[
R_{U}=\text{log}_{2}(1+S_{U})-\sqrt{\frac{V}{n_{F}}}Q^{-1}(\epsilon_{U}^{D}),
\]
where 
\begin{equation}
S_{U}=\frac{\beta^{2}P_{U}}{1+(1+2\alpha^{2})P_{B}}
\end{equation}
is the Signal-to-Interference-plus-Noise Ratio (SINR) for the URLLC
UE, and the channel dispersion $V$ is given as 
\begin{equation}
V=\frac{S_{U}}{1+S_{U}}.
\end{equation}
Notice that in NOMA the incoming URLLC packet is always transmitted,
and hence an URLLC error occurs only if the decoding of such packet
fails, which happens with probability $\epsilon_{U}^{D}$. This implies
that under NOMA, the probability of URLLC error is given by
\begin{equation}
\text{Pr}[E_{U}]=\epsilon_{U}^{D},\label{eq:-35}
\end{equation}
hence imposing the condition $\epsilon_{U}^{D}\leq\epsilon_{U}$ by
the requirement (\ref{eq:-33}).

\subsection{eMBB Rate by Puncturing}

To carry out joint decoding at the BBU of the eMBB signals
under NOMA, the standard approach is to simply discard at the eMBB
decoder those signals that are interfered by URLLC. As shown in Figure
\ref{fig:Non-Orthogonal-Multiple-Access}, this technique, referred
to as puncturing, is based on the detection of URLLC transmissions
at the BBU: if a URLLC transmission is detected in the signal received
from EN$_{k}$, such signal is discarded. Otherwise, the interference-free
eMBB signals are jointly decoded at the BBU. 
\begin{figure}
\begin{centering}
\includegraphics[width=0.5\paperwidth]{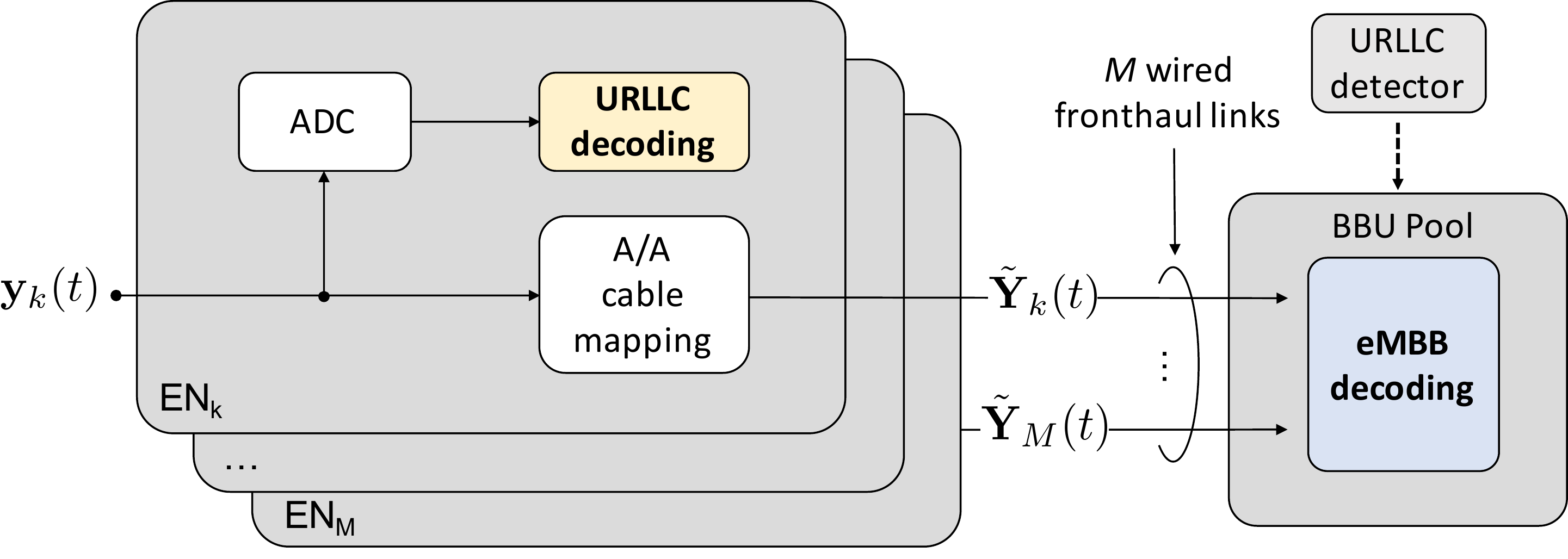}
\par\end{centering}
\caption{\label{fig:Non-Orthogonal-Multiple-Access}Block diagram of the operation
of the ENs and BBU for Non-Orthogonal Multiple Access (NOMA) by puncturing
and Treating Interference as Noise (TIN).  A/A stands for Analog-to-Analog.}
\end{figure}

Considering the aforementioned assumptions, the signal model for puncturing
can be equivalently described by assuming that if the signal $Y_{k}^{f}$
received at the EN$_{k}$ on frequency $f$ is interfered by an URLLC
transmission, then such signal is discarded, and the BBU receives
only noise. This is mathematical described by
\begin{equation}
Y_{k}^{f}=B_{k}(X_{k}^{f}+\alpha X_{k+1}^{f}+\alpha X_{k-1}^{f})+Z_{k}^{f},\label{eq:-19}
\end{equation}
where the Bernoulli variable $B_{k}=1-A_{k}\sim\mathcal{B}(1-q)$
indicates the absence ($B_{k}=1$) or presence ($B_{k}=0$) of URLLC
transmissions in the given minislot. The signal in (\ref{eq:-19})
received across all ENs and frequencies can be written in matrix form
as 
\begin{equation}
\mathbf{Y}=\mathbf{X}\mathbf{H}\mathbf{B}+\mathbf{Z},\label{eq:-24}
\end{equation}
with definitions given in Section \ref{sec:System-Model} and with $\mathbf{B}=\text{diag}(B_{1},B_{2},...,B_{M})$.
The rate for the eMBB UE under NOMA by puncturing is reported in Lemma
\ref{lem:In-the-case-1-1} and can be derived by substituting signal
(\ref{eq:-24}) in Equation (\ref{eq:-28}).
\begin{lem}
\label{lem:In-the-case-1-1}In the given C-RAN architecture with analog
fronthaul links, for a given bandwidth amplification factor $\eta\geq1$,
the eMBB user rate under NOMA by puncturing yields
\begin{equation}
R_{B}=\frac{\mu}{M}\mathbb{E}_{\mathbf{B}}\left[\text{\text{log}}\left(\text{det}\left(\mathbf{I}+P_{B}\mathbf{R}_{z_{\text{eq}}}^{-1}\mathbf{H}_{B,\text{eq}}\mathbf{H}_{B,\text{eq}}^{T}\right)\right)\right],\label{eq:-39}
\end{equation}
where $\mathbf{R}_{z_{\text{eq}}}$ is defined as in Lemma \ref{lem:In-the-case};
$\mathbf{H}_{B,\text{eq}}=\mathbf{B}\mathbf{H}\otimes\mathbf{H}_{c}^{\eta}$
is the equivalent wireless plus cable channel in case of puncturing;
and we have $\mathbf{B}=\text{diag}(B_{1},B_{2},...,B_{M})$ , with
$B_{k}$ being i.i.d. $\mathcal{B}(1-q)$ variables .
\end{lem}
\begin{proof}
Lemma \ref{lem:In-the-case-1-1} can be proved by following similar
steps as for the proof of Lemma 2 with two minor differences: \textit{i)}
the radio channel matrix $\mathbf{H}$ is right multiplied by the
random matrix $\mathbf{B}$ and \textit{ii) }the capacity is computed
by averaging over the distribution of $\mathbf{B}$.
\end{proof}
Differently from the rate (\ref{eq:-36}) achieved by OMA, under NOMA,
the eMBB transmission spreads over all the minislots, so that there
is no scaling factor $1-L_{U}^{-1}$ in front of the rate expression
(\ref{eq:-39}) to account for the resulting loss in spectral efficiency.
In case of NOMA by puncturing, the noise covariance is exactly as
the one in the eMBB OMA rate in Equation (\ref{eq:-36}), since the eMBB
signal, if not discarded at the BBU, is guaranteed to be URLLC interference-free.
The overall rate is computed by averaging over all the possible realizations
of the random matrix $\mathbf{B}$, which left-multiplies the radio
channel matrix $\mathbf{H}$ and accounts for the probability that
the entire signal is discarded due to an incoming URLLC packet.

In the case of C-RAN with digital limited-capacity fronthaul links,
as discussed in \cite{Kassab}, it is advantageous to carry out the
operation of detecting and, eventually, discarding the eMBB signal
at the ENs. In fact, with a digital fronthaul, only the undiscarded
minislots can be quantized, hence devoting the limited fronthaul resources
to increase the resolution of interference-free eMBB samples \cite{Kassab}.
The same does not apply to the analog fronthaul considered here, as
signals are directly relayed to the BBU without any digitization. 

\subsection{eMBB Rate by Treating Interference as Noise}

In the case of analog fronthaul, an enhanced strategy to jointly decode
the eMBB signals under NOMA at the BBU is to treat the URLLC interfering
transmissions as noise at the eMBB decoder, instead of discarding
the corresponding minislot as in puncturing. The block diagram is
the same as for puncturing and shown in Figure \ref{fig:Non-Orthogonal-Multiple-Access}.
Accordingly, based on the signals received over the fronthaul links,
the BBU first detects the presence of URLLC transmission so as to
properly select the decoding metric. Then, based on this knowledge,
joint decoding is performed by TIN.
\begin{lem}
\label{lem:In-the-case-1}In the given C-RAN architecture with analog
fronthaul links, for a given bandwidth amplification factor $\eta\geq1$,
the eMBB user rate under NOMA by treating URLLC interference as noise
yields
\begin{equation}
R_{B}=\frac{\mu}{M}\mathbb{E}_{\mathbf{A}}\left[\text{\text{log}}\left(\text{det}\left(\mathbf{I}+P_{B}\mathbf{R}_{A,z_{\text{eq}}}^{-1}\mathbf{H}_{\text{eq}}\mathbf{H}_{\text{eq}}^{T}\right)\right)\right],\label{eq:-30}
\end{equation}
where $\mathbf{R}_{A,z_{\text{eq}}}=\mathbf{R}_{z_{\text{eq}}}+\beta^{2}P_{U}\left(\mathbf{A}\otimes\mathbf{H}_{c}^{\eta}\mathbf{H}_{c}^{\eta}\right)$
is the overall noise plus URLLC interference at the BBU; matrix $\mathbf{A}$
is as in (\ref{eq:-14}); and matrices $\mathbf{R}_{z_{\text{eq}}}$
and $\mathbf{H}_{\text{eq}}$ are the same as in Lemma \ref{lem:In-the-case}.
\end{lem}
\begin{proof}
see Appendix \ref{subsec:Proof-of-Lemma}.
\end{proof}
Differently from the two previous cases, in the case of NOMA under TIN,
the noise covariance matrix $\mathbf{R}_{A,z_{\text{eq}}}$ needs
to account also for the interfering URLLC transmissions, whose packet
arrival probability is described by matrix $\mathbf{A}$. The achievable
rate is then computed by taking the average over the random matrix
$\mathbf{A}.$ This average reflects the long-blocklength transmissions
of the eMBB users.

\subsection{eMBB Rate by Successive Interference Cancellation}

Finally, a more complex receiver architecture can be considered at
the BBU, whereby interference is cancelled out from the useful signal.
This technique, referred to as Successive Interference Cancellation (SIC), is based on the idea that, if an URLLC signal is successfully
decoded at the EN$_{k}$, it can be cancelled from the overall received
signal $\mathbf{y}_{k}$ prior to the relaying over the cable, so
that an ideally interference-free eMBB signal is forwarded to the
BBU. We also assume that, if the URLLC signal is not successfully
decoded, signal $\mathbf{y}_{k}$ is discarded. 

As a practical note, SIC must be performed in the analog domain, thus
complicating the system design. Practical complications are not considered
in the analysis here. As shown in Figure \ref{fig:Non-Orthogonal-Multiple-Access-2},
if the URLLC signal is successfully decoded at EN$_{k}$, this needs
first to be Digital-to-Analog Converted (DAC) and then cancelled from
the analog signal $\mathbf{y}_{k}$. Therefore, signal $\mathbf{y}_{k}$
needs to be suitably delayed in order to wait for the cascade of ADC,
decoding, and DAC operations to be completed at the URLLC decoder.
Being latency not an issue for eMBB traffic, in this work we assume
to employ ideal ADC/DAC, so that the delay in Figure \ref{fig:Non-Orthogonal-Multiple-Access-2}
is assumed as ideally zero. 
\begin{figure}
\begin{centering}
\includegraphics[width=0.65\paperwidth]{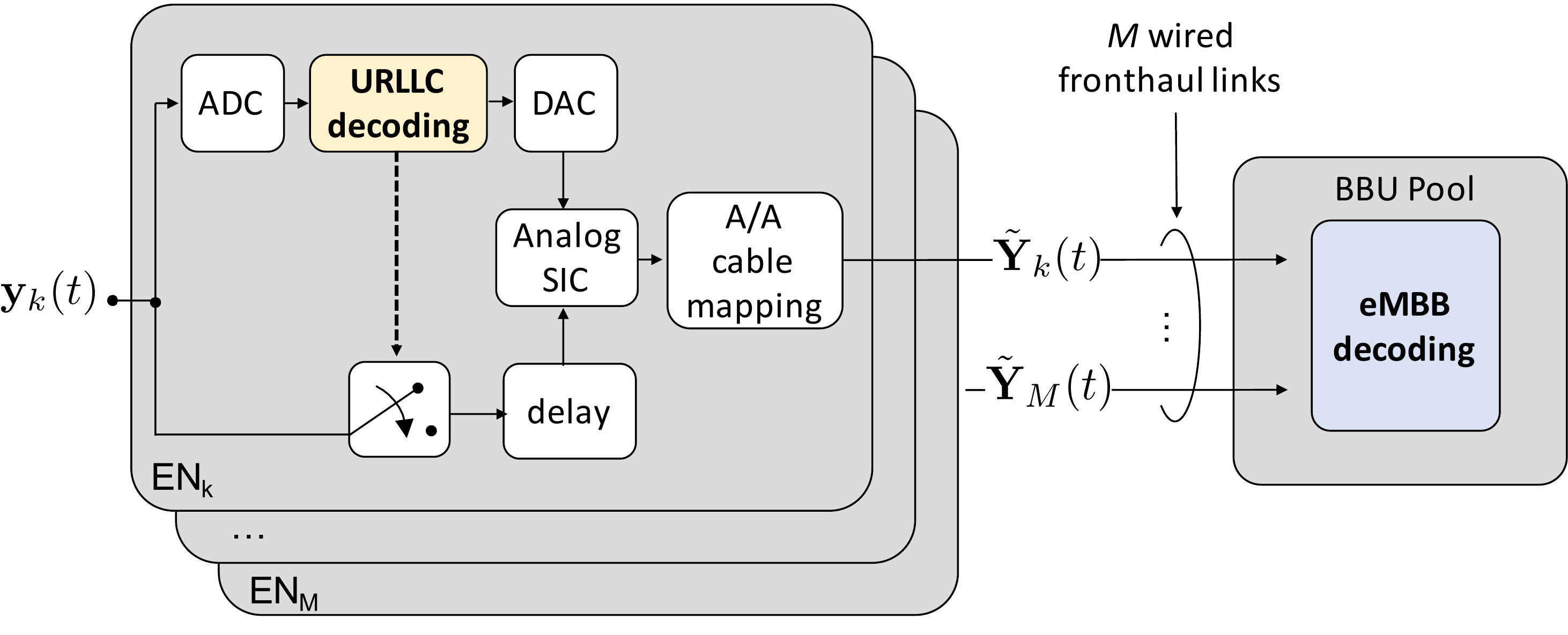}
\par\end{centering}
\caption{\label{fig:Non-Orthogonal-Multiple-Access-2}Block diagram of the
operation of the ENs and BBU for Non-Orthogonal Multiple Access (NOMA)
by Successive Interference Cancellation (SIC).  A/A stands for Analog-to-Analog.}

\end{figure}

To account for imperfect SIC, the amplitude of the residual
URLLC interference on eMBB signal is assumed to be proportional to
a factor $\rho\in[0,1]$. Accordingly, perfect SIC corresponds to
$\rho=0$, and no SIC to $\rho=1$.

The signal received at the BBU from EN$_{k}$ can be thus written
as 
\begin{equation}
Y_{k}^{f}=(1-A_{k}E_{k})(X_{k}^{f}+\alpha X_{k+1}^{f}+\alpha X_{k-1}^{f})+\rho\beta A_{k}(1-E_{k})U_{k}^{f}+Z_{k}^{f},\label{eq:-25}
\end{equation}
where the Bernoulli variable $E_{k}\sim\mathcal{B}(q\epsilon_{U}^{D})$
indicates whether there has been an error in decoding the URLLC packet
(i.e., $E_{k}=1$), or it has been successfully decoded (i.e., $E_{k}=0$),
and $A_{k}$ is the same as above. It is easy to show that the factor
$(1-A_{k}E_{k})$ multiplying the eMBB signal indicates that the eMBB
signal (\ref{eq:-25}) is discarded only when the two following events
simultaneously happen: \textit{i)} there is a URLLC transmission (i.e.,
$A_{k}=1$, whose probability is $q$), and \textit{ii)} such URLLC
transmission is not successfully decoded ($E_{k}=1$, whose probability
is $\epsilon_{U}^{D}$). In turn, the factor $A_{k}(1-E_{k})$ multiplying
the URLLC signal implies that, if there is a URLLC transmission (i.e.,
$A_{k}=1$) and such transmission is successfully decoded at the EN$_{k}$
(i.e., $E_{k}=0$), then the URLLC signal is mitigated by analog SIC
so that only a $\rho$-fraction of it is forwarded to the BBU and
impairs the eMBB transmission.

The signal in (\ref{eq:-25}) received across all ENs and frequencies
in the case of NOMA by SIC can be equivalently written in matrix form
as 
\begin{equation}
\mathbf{Y}=\mathbf{X}\mathbf{H}\left(\mathbf{I}-\mathbf{A}\mathbf{E}\right)+\rho\beta\mathbf{U}\mathbf{A}\left(\mathbf{I}-\mathbf{E}\right)+\mathbf{Z},\label{eq:-26}
\end{equation}
where $\mathbf{E}=\text{diag}(E_{1},E_{2},...,E_{M})$. The eMBB UE
rate for NOMA by SIC can thus be computed by substituting signal (\ref{eq:-26})
in (\ref{eq:-28}) and the final result is in Lemma \ref{lem:In-the-case-1-1-1}.
\begin{lem}
\label{lem:In-the-case-1-1-1}In the given C-RAN architecture with
analog fronthaul links, for a given bandwidth amplification factor
$\eta\geq1$, the eMBB user rate under NOMA by SIC yields
\begin{equation}
R_{B}=\frac{\mu}{M}\mathbb{E}_{\mathbf{A},\mathbf{E}}\left[\text{\text{log}}\left(\text{det}\left(\mathbf{I}+P_{B}\mathbf{R}_{AE,z_{\text{eq}}}^{-1}\mathbf{H}_{AE}\mathbf{H}_{AE}^{T}\right)\right)\right],\label{eq:-29}
\end{equation}
where $\mathbf{H}_{AE}=\left(\mathbf{(I-AE)}\mathbf{H}\right)\otimes\mathbf{H}_{c}^{\eta}$
is the equivalent wireless plus cable channel in case of SIC; $\mathbf{E}=\text{diag}(E_{1},E_{2},...,E_{M})$
is a diagonal matrix whose $k$-th entry $E_{k}\sim\mathcal{B}(q\epsilon_{U}^{D})$
accounts for the probability that the URLLC signal is not successfully
decoded at EN$_{k}$; and $\mathbf{R}_{AE,z_{\text{eq}}}=\mathbf{R}_{z_{\text{eq}}}+\rho^{2}\beta^{2}P_{U}\left(\left(\mathbf{A}(\mathbf{I}-\mathbf{E})\right)\otimes\mathbf{H}_{c}^{\eta}\mathbf{H}_{c}^{\eta}\right)$
is the overall noise plus residual URLLC interference.
\end{lem}
\begin{proof}
Lemma 5 can be proved by following similar steps as for the proofs
of the previous Lemmas.
\end{proof}
SIC describes a more complex ENs architecture in which the URLLC signals,
if successfully decoded, are successively canceled from the eMBB signals
at the EN. However, in the case of imperfect interference cancellation,
i.e., $\rho>0$, the eMBB signal is still impaired by some residual
URLLC interference, which is accounted for by the overall noise covariance
$\mathbf{R}_{AE,z_{\text{eq}}}$ in (\ref{eq:-29}), similarly to
TIN. The URLLC arrival probability and the probability of successful
decoding of the URLLC packets are reflected by random matrices $\mathbf{A}$
and $\mathbf{E}$, respectively.

\section{Numerical Results\label{sec:Numerical-Results}}

Numerical results based on the previous theoretical discussion are
shown in this section with the aim of providing some useful intuitions
about the performance of C-RAN systems based on analog RoC in the
presence of both URLLC and eMBB services. Unless otherwise stated,
we consider the following settings: $M=6$ ENs, $n_{F}=60$ subcarriers ({This choice is motivated by the fact that, while still resembling
the properties of URLLC short-packet transmissions, $n_{F}=60$ is
a sufficiently long packet size to ensure tight lower and upper bounds
for the limited-blocklength channel capacity \cite{Durisi-Koch-Popovski_2016}}), $P_{B}=7\,$dB, $P_{U}=10\,$dB, URLLC channel gain $\beta^{2}=1$,
$P_{C}=7\,$dB, $l_{S}=4$, and, conventionally, $\epsilon_{U}=10^{-3}$.
In the case of OMA, the worst-case access latency for URLLC users
is set to $L_{U}=2$ minislots.

Figure \ref{fig:URLLC-and-eMBB} shows URLLC and eMBB per-UE rates for
both OMA and NOMA by varying the fronthaul crosstalk interference
power $\gamma^{2}$. We consider two values for the normalized bandwidth
$\mu$ of each copper cable, namely $\mu=1/l_{S}=1/4$ and $\mu=1$.
Please note that the first value corresponds to the minimal bandwidth, while
the latter enables each twisted-pair to carry the whole signal bandwidth.
For reference, the eMBB rates obtained in the case of ideal fronthaul
are shown for both OMA and NOMA. The inter-cell interference power
is set to $\alpha^{2}=0.2$, and the URLLC arrival probability to
$q=10^{-3}$. For NOMA, we consider here only puncturing.

The URLLC rates do not depend on cable crosstalk interference $\gamma^{2}$,
since URLLC packets are decoded at the EN and thus never forwarded
to the BBU over the fronthaul. Furthermore, with NOMA, the access
latency of URLLC is minimum, i.e., $L_{U}=1$, while for OMA it equals
$L_{U}=2$. However, Figure \ref{fig:URLLC-and-eMBB} shows that the
price to pay for this reduced latency is in terms of transmission
rate, which is lower than in the OMA case. On the contrary, in the
case of eMBB, NOMA allows for a communication at higher rates than
those achieved by OMA, thanks to the larger available bandwidth when
$q$ is small~enough. 
\begin{figure}
\begin{centering}
\includegraphics[width=0.8\columnwidth]{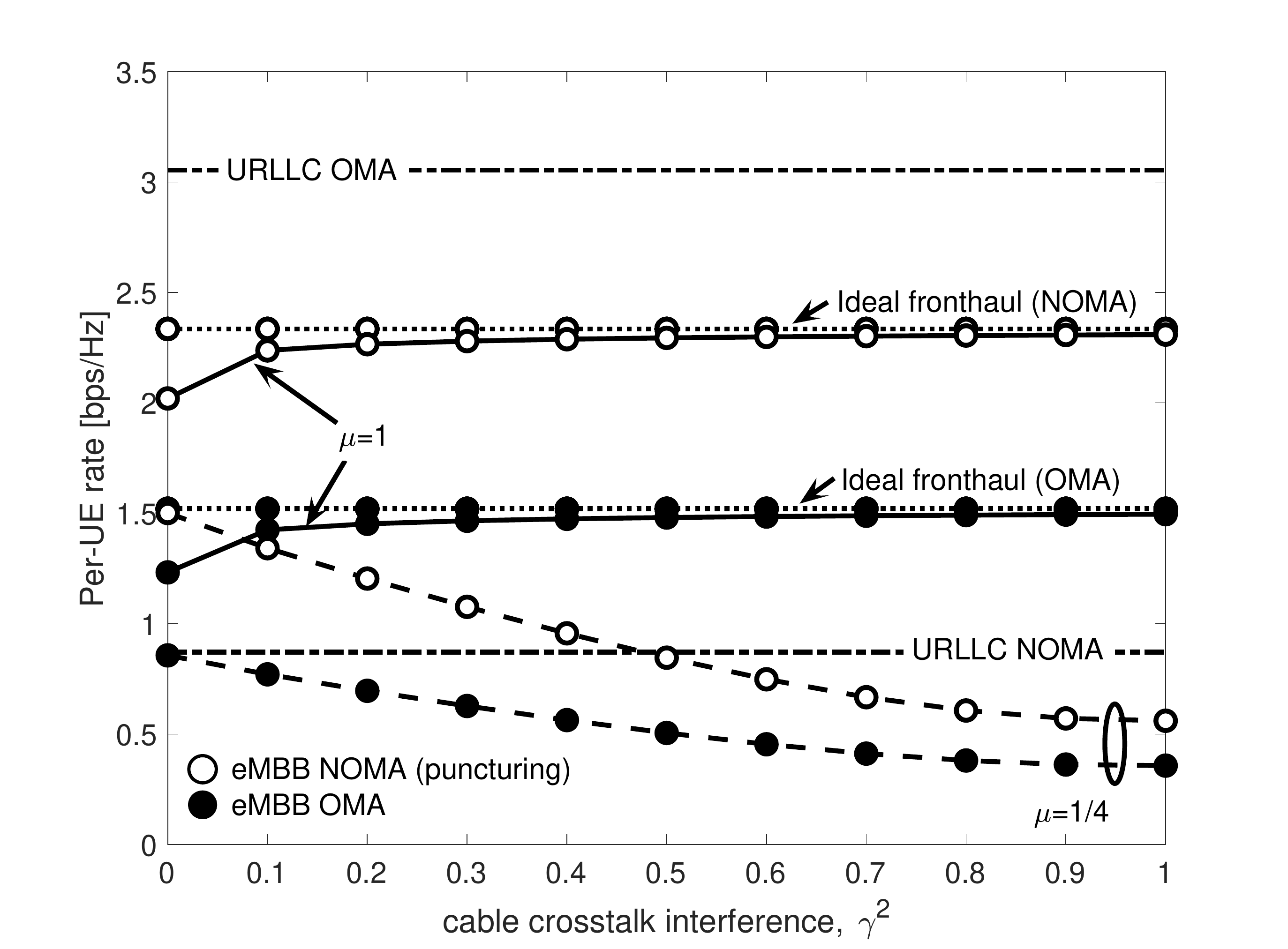}
\par\end{centering}
\caption{\label{fig:URLLC-and-eMBB}URLLC and eMBB per-UE rates as a function
of fronthaul crosstalk interference power $\gamma^{2}$ for OMA and
NOMA with puncturing.}
\end{figure}

On the subject of eMBB rates, it is interesting to discuss the interplay
between the normalized cable bandwidth $\mu$ and crosstalk power
$\gamma^{2}$. For $\mu=1$, the same signal is transmitted over all
the $l_{S}=4$ fronthaul twisted-pairs, and hence the four spatial
paths sum coherently over the cable, thus turning crosstalk into
a benefit. Hence, the eMBB rates under both OMA and NOMA increase
with $\gamma^{2}$ and ultimately converge to those achieved over
the ideal fronthaul ({In this work, we consider the same interference gain $\gamma$ for all fronthaul links. In practice, the performance boost shown in Figure 10 for increasing $\gamma$ and for $\mu=1$ would still be present, albeit to a different extent dependent on the channel realization, even if considering complex channel gain. However, this would require the use of more complex precoding techniques, such as Tomlinson-Harashima \cite{Hekrdla-Matera-etal_2015}, which require an estimate of the fronthaul channel. Since for wired fronthauling the channel is nearly static and time-invariant, channel state information can be easily obtained \cite{zafaruddin2017signal}}). For $\mu=1/4$ instead, disjoint portions of
the radio signal are transmitted over different interfering twisted-pairs,
and the performance progressively decrease with $\gamma^{2}$. The
leftmost portion of Figure \ref{fig:URLLC-and-eMBB} suggests that for
mild cable interference, even when the cable bandwidth is small (i.e.,
$\mu=1/4$), it is still possible to provide communication with acceptable
performance degradation, i.e., with a $\approx1$ bps/Hz loss from
the ideal fronthaul case for NOMA, and an even smaller loss for OMA.
However, when the cable crosstalk increases, the rate degradation
is severe, and fair performance are achieved only if the cable bandwidth
is large enough to accommodate the redundant transmission of radio
signals over all pairs, i.e., $\mu=1$.

Figure \ref{fig:URLLC-and-eMBB-1} shows URLLC and eMBB rates as a function
of the URLLC packet arrival probability $q$. eMBB rates under OMA
are compared with those achieved by NOMA under puncturing, TIN and
SIC. We consider here full cable bandwidth availability $\mu=1$ and
$\gamma^{2}=1$, so that, as in Figure \ref{fig:URLLC-and-eMBB}, the
rates achieved for both OMA and NOMA for low $q$ (say, $q<10^{-2}$)
coincide with those achieved over the ideal fronthaul. Inter-cell
interference is set to $\alpha^{2}=0.2$. As noted in \cite{Kassab},
under OMA,when $q$ increases, the probability of an URLLC packet
to be dropped due to blockage becomes very high, preventing URLLC
transmission from meeting the strict reliability constraints, and
results in a vanishing URLLC rate. This is unlike in NOMA, whereby
the URLLC rate is not affected by $q$, and the access latency is
minimal, i.e., $L_{U}=1$. For eMBB under NOMA, TIN always outperforms
puncturing. This is because TIN does not discard any received minislot,
thus contributing to the overall eMBB rate. The result is in contrast
with the conventional digital capacity-constrained fronthaul considered
in \cite{Kassab}. In fact, in the latter, for sufficient low $q$,
it is preferable not to waste fronthaul capacity resources by quantizing
samples received in minislots affected by URLLC interference in order
to increase the resolution of the interference-free samples. Additional
gains are achieved by SIC, which takes advantages of the high reliability,
and thus high probability to be cancelled, of the URLLC signal at
the EN.
\begin{figure}
\begin{centering}
\includegraphics[width=0.8\columnwidth]{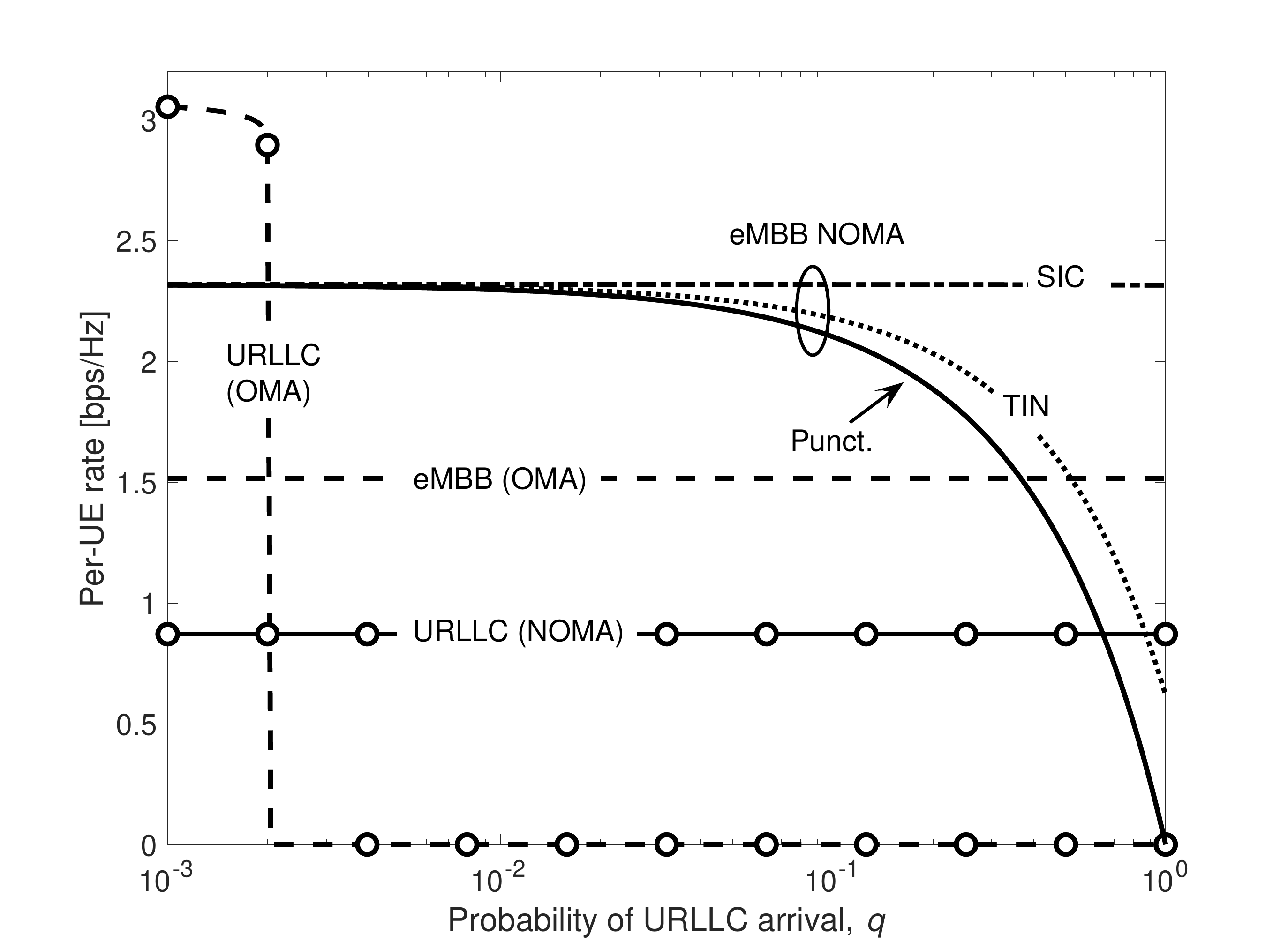}
\par\end{centering}
\caption{\label{fig:URLLC-and-eMBB-1}URLLC and eMBB rates vs probability of
URLLC arrival $q$ for OMA and NOMA by puncturing, treating interference
as noise (TIN), and successive interference cancellation (SIC). }
\end{figure}

Implementing SIC in a fully analog fashion is practically not trivial,
and there is generally some residual URLLC interference. The effect
of residual interference on the achievable eMBB rate is investigated
in Figure \ref{fig:eMBB-rates-for} for $q=0.3$, $\alpha^{2}=0.4$,
$\gamma^{2}=0.5,$ $\mu=1$ and for different power of URLLC UE $P_{U}$.
Once again, in case of perfect interference cancellation, i.e., $\rho=0$,
SIC approaches the ideal fronthaul performance, while for more severe
values of the residual interference power $\rho$, the eMBB performance
progressively decreases. Nevertheless, even in the worst-case SIC
scenario, i.e., $\rho=1$, the achievable rates are never worse
than those achieved by TIN irrespective of the value of $P_{U}$.
This is once again due to the high reliability of URLLC transmission.
It is in fact easy to prove that for $\rho=1$ and low values of $\epsilon_{U}^{D}$,
the SIC eMBB rate in (\ref{eq:-29}) converges to the one of TIN in
(\ref{eq:-30}).
\begin{figure}
\begin{centering}
\includegraphics[width=0.6\columnwidth]{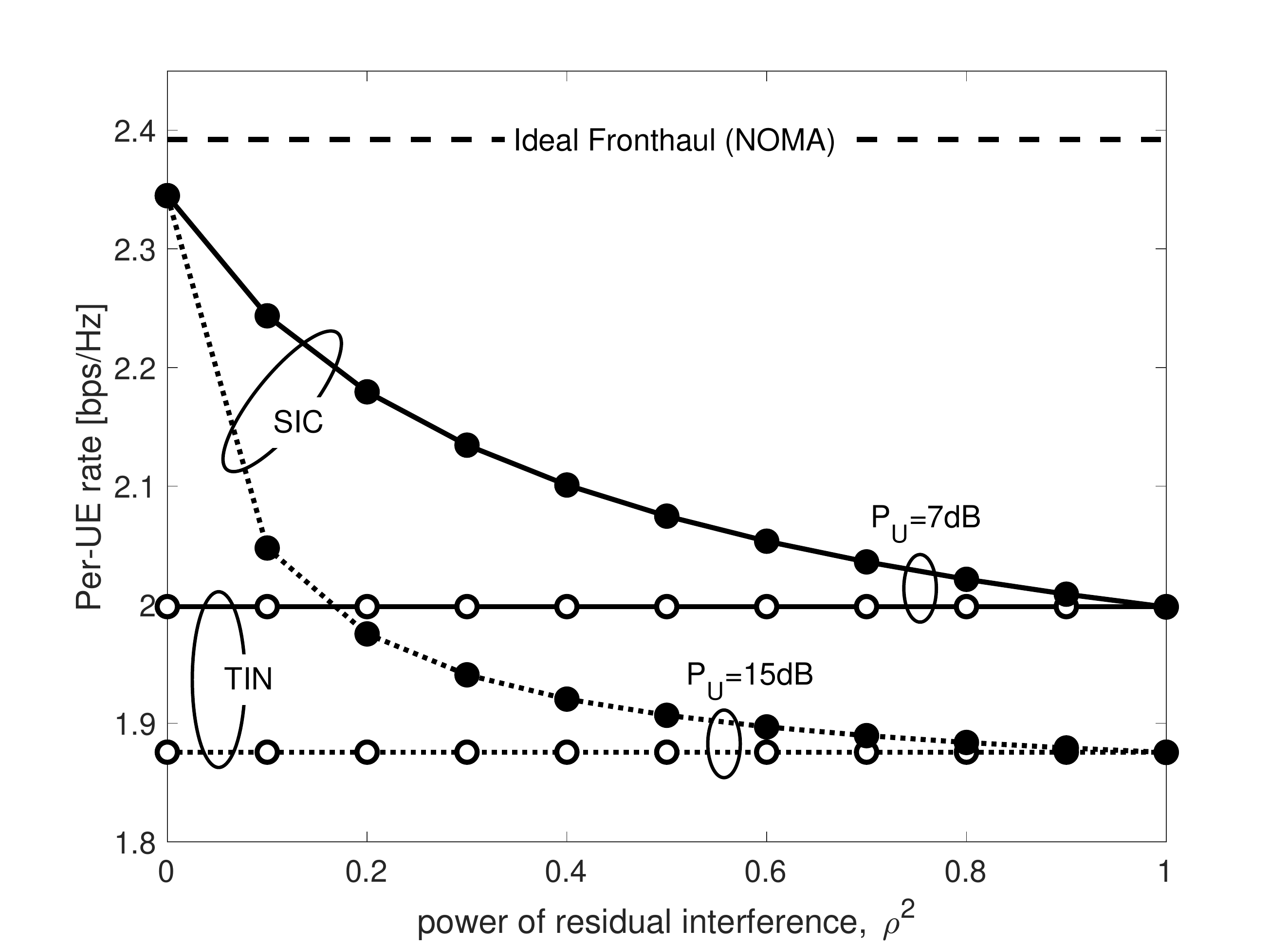}
\par\end{centering}
\caption{\label{fig:eMBB-rates-for}eMBB rates for NOMA by SIC vs power of
residual URLLC interference $\rho^{2}$.}
\end{figure}

For completeness, Figure \ref{fig:URLLC-and-eMBB-2} shows the trade-off
between eMBB and URLLC per-UE rates as a function of the access latency
$L_{U}$ for OMA and NOMA with puncturing, and considering $q=10^{-3}$,
$\alpha^{2}=0.2$ and $\gamma^{2}=0.5$. The behavior of the RoC-based
C-RAN system versus the access latency $L_{U}$ is similar to the
one observed for digital capacity-constrained fronthaul \cite{Kassab} for
both $\mu=1$ and $\mu=1/4$. While under OMA it is not possible to
achieve a non-zero URLLC rate at even relatively low access latency such
as $L_{U}>3$, NOMA provides a reliable communication with constant
minimal $L_{U}=1$ access latency, but with lower rate. For eMBB,
NOMA achieves an higher per-UE rate regardless of the value of the
normalized bandwidth $\mu$.
\begin{figure}
\begin{centering}
\includegraphics[width=0.6\columnwidth]{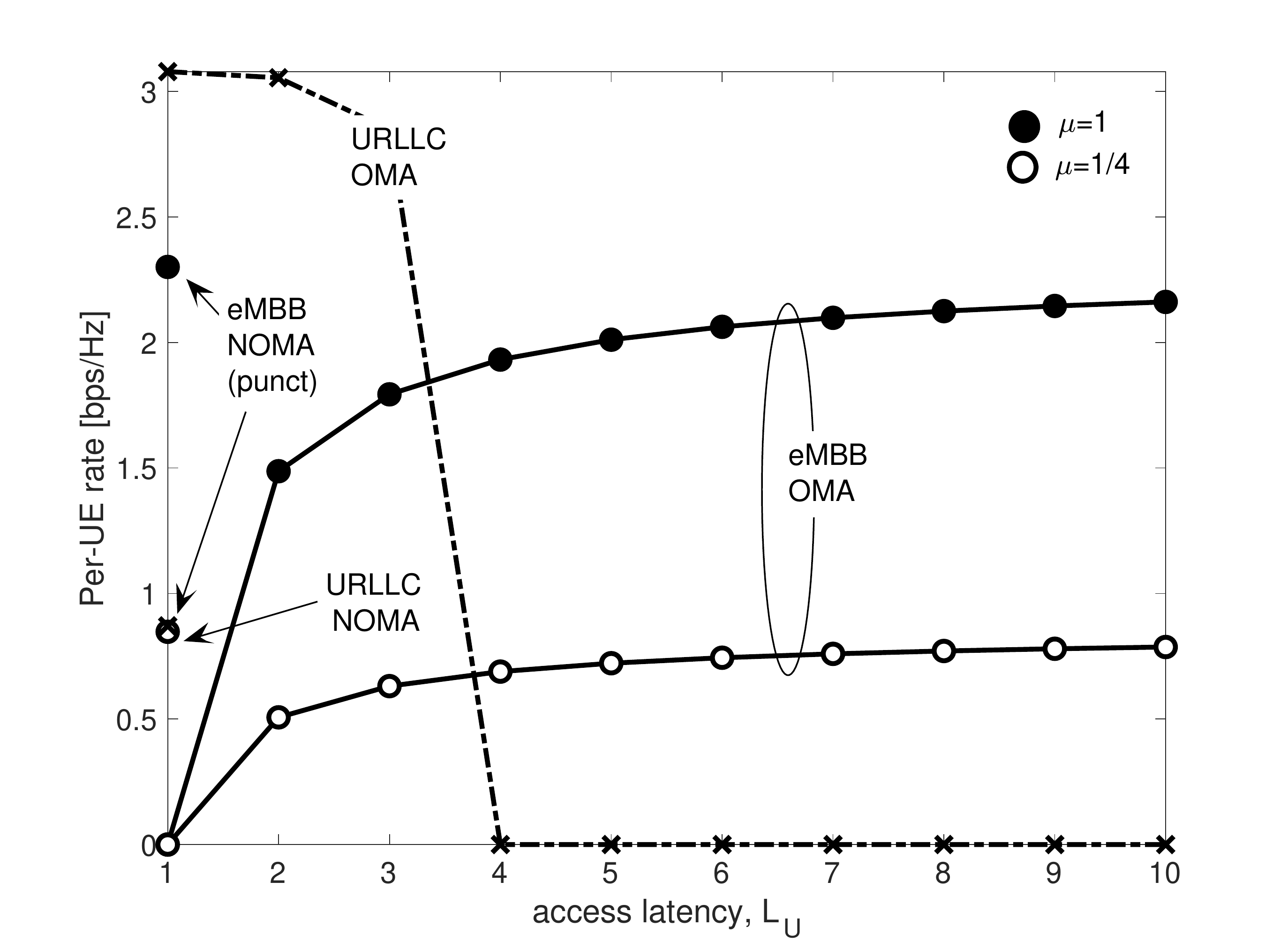}
\par\end{centering}
\caption{\label{fig:URLLC-and-eMBB-2}URLLC and eMBB per-UE rates as a function
of access latency $L_{U}$ for OMA and NOMA with puncturing.}
\end{figure}

\section{Conclusions\label{sec:Conclusions}}

This paper considers the coexistence of eMBB and URLLC services in
the uplink of an analog C-RAN architecture from an information theoretic
perspective. The rate expressions for URLLC and eMBB users under Orthogonal
and Non-Orthogonal Multiple Access (OMA and NOMA, respectively) have
been derived considering Analog Radio-over-Copper (A-RoC) as a sample
scenario, although the proposed model can be easily adapted to other
analog fronthaul technologies. For eMBB signals, performance have
been evaluated in terms of information rate, while for URLLC we also
took into account worst-case access latency and reliability. In case
of NOMA, different decoding strategies have been considered in order
to mitigate the impact of URLLC transmission on eMBB information rate.
In particular, the performance achieved by puncturing, considered for 5G standardization, have been compared with those achieved
by Treating URLLC Interference as Noise (TIN), and by Successive URLLC
Interference Cancellation (SIC). 

The analysis showed that NOMA allows for higher eMBB information rates
with respect to OMA, while guaranteeing a reliable low-rate URLLC
communication with minimal access latency. Furthermore, numerical
results demonstrated that, differently from the digital C-RAN architecture
based on limited-capacity fronthaul links, for analog C-RAN, TIN always
outperforms puncturing, and SIC achieves the best performance at the
price of an higher decoder complexity. 

As work in progress, the theoretical model can be extended to account
for fading channels or geometric mmWave-link channel models. 

Similarly, a frequency-dependent cable channel can be considered by
making the cable crosstalk coefficient $\gamma$ increase with cable
frequency \cite{Naqvi-Matera-etal_2017}. Another interesting research
direction is to consider the case in which the BBU has no knowledge
about the incoming signal, i.e., it is not able to detect the URLLC
transmissions, so that it is impossible for the BBU to choose the
proper metric for joint signal decoding \cite{zhang2012general,merhav1994information}.
Finally, the overall system can be extended to the case of multiple
users per-cell, where both ENs and users are equipped with multiple
antennas.

\vspace{5px}

\appendices

\section{Proof of Lemma 1\label{subsec:Example:-Signal-Combining}}

The proof of Lemma \ref{lem:In-C-RAN-architecture} is structured
in two main steps detailed in the following: \textit{(i)} MRC at the
cable output, \textit{(ii)} cable signal vectorization.

\subsection{Maximum Ratio Combining at the Cable Output}

The signal $\mathbf{R}_{k}\in\mathbb{C}^{l_{F}\times\frac{1}{\mu}}$
after the MRC at the cable output is
\begin{equation}
\begin{aligned}\mathbf{R}_{k} & =\tilde{\mathbf{R}}_{k}\mathbf{G}\\
 & =\left(\tilde{\mathbf{Y}}_{k}\mathbf{H}_{c}+\tilde{\mathbf{W}}_{k}\right)\mathbf{G},
\end{aligned}
\label{eq:-7}
\end{equation}
where the definitions of $\tilde{\mathbf{R}}_{k}$, $\mathbf{G}$,
$\tilde{\mathbf{Y}}_{k}$ and $\mathbf{H}_{c}$ follow from and (\ref{eq:-1}),
(\ref{eq:-20}), (\ref{eq:-5}) and (\ref{eq:-21}), respectively.
Thus, Equation (\ref{eq:-7}) can be rewritten as

\begin{equation}
\begin{aligned}\mathbf{R}_{k} & =\left(\mathbf{Y}_{k}\otimes\mathbf{1}_{\eta}^{T}\right)\left(\gamma\mathbf{1}_{l_{S}}\mathbf{1}_{l_{S}}^{T}+(1-\gamma)\mathbf{I}_{l_{S}}\right)\mathbf{G}+\tilde{\mathbf{W}}_{k}\mathbf{G}\\
 & =\underset{\mathbf{R}_{k}^{'}}{\underbrace{\gamma\left(\mathbf{Y}_{k}\otimes\mathbf{1}_{\eta}^{T}\right)\mathbf{1}_{l_{S}}\mathbf{1}_{l_{S}}^{T}\mathbf{G}}}+\underset{\mathbf{R}_{k}^{''}}{\underbrace{(1-\gamma)\left(\mathbf{Y}_{k}\otimes\mathbf{1}_{\eta}^{T}\right)\mathbf{G}}}+\underset{\mathbf{W}_{k}}{\underbrace{\tilde{\mathbf{W}}_{k}\mathbf{G}}}
\end{aligned}
,\label{eq:-8}
\end{equation}
where $\mathbf{W}_{k}=\tilde{\mathbf{W}}_{k}\mathbf{G}$ is the noise
post MRC. The first term in (\ref{eq:-8}) can be rewritten as
\begin{equation}
\begin{aligned}\mathbf{R}_{k}^{'} & =\frac{\gamma}{\eta}\left(\mathbf{Y}_{k}\otimes\mathbf{1}_{\eta}^{T}\right)\mathbf{1}_{l_{S}}\mathbf{1}_{l_{S}}^{T}\left(\mathbf{I}_{\frac{1}{\mu}}\otimes\mathbf{1}_{\eta}\right)\\
 & \overset{^{(a)}}{=}\gamma\left(\mathbf{Y}_{k}\otimes\mathbf{1}_{\eta}^{T}\right)\mathbf{1}_{l_{S}}\mathbf{1}_{\frac{1}{\mu}}^{T}\\
 & \overset{^{(b)}}{=}\gamma\eta\mathbf{Y}_{k}\mathbf{1}_{\frac{1}{\mu}}\mathbf{1}_{\frac{1}{\mu}}^{T},
\end{aligned}
\label{eq:-9}
\end{equation}
where $\overset{^{(a)}}{=}$ comes from the fact that $\mathbf{1}_{l_{S}}^{T}\left(\mathbf{I}_{\frac{1}{\mu}}\otimes\mathbf{1}_{\eta}\right)=\left(\mathbf{1}_{\frac{1}{\mu}}^{T}\otimes\mathbf{1}_{\eta}^{T}\right)\left(\mathbf{I}_{\frac{1}{\mu}}\otimes\mathbf{1}_{\eta}\right)=\eta\mathbf{1}_{\frac{1}{\mu}}^{T}$
due to the mixed-product property of Kronecker product operator $\left(\mathbf{A}\otimes\mathbf{B}\right)\left(\mathbf{C}\otimes\mathbf{D}\right)=\mathbf{A}\mathbf{C}\otimes\mathbf{B}\mathbf{D}$
(see \cite{spagnolini2018statistical}, Chapter 2), and, similarly,
$\overset{^{(b)}}{=}$ is obtained by $\left(\mathbf{Y}_{k}\otimes\mathbf{1}_{\eta}^{T}\right)\mathbf{1}_{l_{S}}\mathbf{1}_{\frac{1}{\mu}}^{T}=\left(\mathbf{Y}_{k}\otimes\mathbf{1}_{\eta}^{T}\right)\left(\mathbf{1}_{\frac{1}{\mu}}\mathbf{1}_{\frac{1}{\mu}}^{T}\otimes\mathbf{1}_{\eta}\right)=\eta\mathbf{Y}_{k}\mathbf{1}_{\frac{1}{\mu}}\mathbf{1}_{\frac{1}{\mu}}^{T}$.
Using similar arguments, the second term in (\ref{eq:-8}) simplifies
to
\begin{equation}
\begin{aligned}\mathbf{R}_{k}^{''} & =\frac{1-\gamma}{\eta}\left(\mathbf{Y}_{k}\otimes\mathbf{1}_{\eta}^{T}\right)\left(\mathbf{I}_{\frac{1}{\mu}}\otimes\mathbf{1}_{\eta}\right)\\
 & =(1-\gamma)\mathbf{Y}_{k}.
\end{aligned}
\label{eq:-10}
\end{equation}
Finally, substituting (\ref{eq:-9}) and (\ref{eq:-10}) in (\ref{eq:-8})
we obtain 
\begin{equation}
\mathbf{R}_{k}=\mathbf{Y}_{k}\mathbf{H}_{c}^{\eta}+\mathbf{W}_{k},\label{eq:-11}
\end{equation}
where 
\begin{equation}
\mathbf{H}_{c}^{\eta}=\gamma\eta\mathbf{1}_{\frac{1}{\mu}}\mathbf{1}_{\frac{1}{\mu}}^{T}+(1-\gamma)\mathbf{I}_{\frac{1}{\mu}}
\end{equation}
 is the equivalent cable channel matrix accounting for the bandwidth
amplification factor over cable $\eta$, and $\mathbf{Y}_{k}$ is
the radio signal reorganized in matrix form as in (\ref{eq:}).

\subsection{Cable Signal Vectorization}

The vector signal $\mathbf{r}_{k}\in\mathbb{C}^{n_{F}\times1}$ received
at the BBU from the $k$-th EN over all the radio frequency channels
can be obtained by vectorizing matrix $\mathbf{R}_{k}$ in (\ref{eq:-11})
as 
\begin{equation}
\begin{aligned}\mathbf{r}_{k} & =\text{vec}(\mathbf{Y}_{k}\mathbf{H}_{c}^{\eta}+\mathbf{W}_{k})\\
 & =\left(\mathbf{H}_{c}^{\eta}\otimes\mathbf{I}_{l_{F}}\right)\text{vec}(\mathbf{Y}_{k})+\text{vec}(\mathbf{W}_{k})\\
 & =\left(\mathbf{H}_{c}^{\eta}\otimes\mathbf{I}_{l_{F}}\right)\mathbf{y}_{k}+\mathbf{w}_{k},
\end{aligned}
\label{eq:-13}
\end{equation}
where $\mathbf{y}_{k}$ is the radio signal received at EN $k$-th.
The overall cable noise vector $\mathbf{w}_{k}$ can be rewritten
as
\begin{equation}
\begin{aligned}\mathbf{w}_{k} & =\text{vec}(\mathbf{W}_{k})\\
 & =\text{vec}(\tilde{\mathbf{W}}_{k}\mathbf{G})\\
 & =\left(\mathbf{G}^{T}\otimes\mathbf{I}_{l_{F}}\right)\text{vec}(\tilde{\mathbf{W}}_{k})\\
 & =\frac{1}{\eta}\left(\mathbf{I}_{\frac{1}{\mu}}\otimes\mathbf{1}_{\eta}^{T}\otimes\mathbf{I}_{l_{F}}\right)\tilde{\mathbf{w}}_{k}.
\end{aligned}
\end{equation}
It is important to notice that since the cable noise $\tilde{\mathbf{W}}_{k}$
is white Gaussian and uncorrelated over cable pairs (see (\ref{eq:-1})),
$\tilde{\mathbf{w}}_{k}=\text{vec}(\tilde{\mathbf{W}}_{k})$ is also
white Gaussian distributed as $\tilde{\mathbf{w}}_{k}\sim\mathcal{CN}(\mathbf{0},\mathbf{I}_{l_{S}l_{F}})$.
Hence, the covariance $\mathbf{R}_{w}$ of the overall cable noise
vector $\mathbf{w}_{k}$ yields
\begin{equation}
\begin{aligned}\mathbf{R}_{w} & =\mathbb{E}\left[\mathbf{w}_{k}\mathbf{w}_{k}^{H}\right]\\
 & =\frac{1}{\eta^{2}}\left(\mathbf{I}_{\frac{1}{\mu}}\otimes\mathbf{1}_{\eta}^{T}\otimes\mathbf{I}_{l_{F}}\right)\left(\mathbf{I}_{\frac{1}{\mu}}\otimes\mathbf{1}_{\eta}^{T}\otimes\mathbf{I}_{l_{F}}\right)^{T}\\
 & \overset{^{(a)}}{=}\frac{1}{\eta}\left(\mathbf{I}_{\frac{1}{\mu}}\otimes\mathbf{I}_{l_{F}}\right)\\
 & =\frac{1}{\eta}\mathbf{I}_{n_{F}},
\end{aligned}
\label{eq:-12}
\end{equation}
where the equality $^{(a)}$ comes again from the mixed-product property
of Kronecker product. Equation (\ref{eq:-12}) shows that the MRC allows
to take advantages from the signal redundancy over the cable, which
results in a reduction of the cable noise power by a factor of $\eta$.

The proof is completed by gathering the signals $\mathbf{r}_{k}$
(\ref{eq:-13}) received at the BBU from all ENs as
\begin{equation}
\begin{aligned}\mathbf{R} & =\left[\mathbf{r}_{1},\mathbf{r}_{2},...,\mathbf{r}_{M}\right]\\
 & =\left(\mathbf{H}_{c}^{\eta}\otimes\mathbf{I}_{l_{F}}\right)[\mathbf{y}_{1},\mathbf{y}_{2},...,\mathbf{y}_{M}]+[\mathbf{w}_{1},\mathbf{w}_{2},...,\mathbf{w}_{M}]\\
 & =\left(\mathbf{H}_{c}^{\eta}\otimes\mathbf{I}_{l_{F}}\right)\mathbf{Y}+\mathbf{W},
\end{aligned}
\label{eq:-16}
\end{equation}
where $\mathbf{Y}$ is the signal received by all ENs over all radio
channels in (\ref{eq:-14}). 

\section{Proof of Lemma 2 }\label{Lemma2}

By substituting signal $\mathbf{Y}$ in (\ref{eq:-15-1}) received
by all ENs over all radio channels in case of OMA in Equation (\ref{eq:-28}),
we obtain
\begin{equation}
\mathbf{R}=\left(\mathbf{H}_{c}^{\eta}\otimes\mathbf{I}_{l_{F}}\right)\mathbf{X}\mathbf{H}+\left(\mathbf{H}_{c}^{\eta}\otimes\mathbf{I}_{l_{F}}\right)\mathbf{Z}+\mathbf{W}.
\end{equation}
To compute the per-UE eMBB rate, a further vectorization
is needed, leading to
\begin{equation}
\begin{aligned}\mathbf{r}_{\text{}} & =\text{vec}(\mathbf{R})\\
 & =\text{vec}\left(\left(\mathbf{H}_{c}^{\eta}\otimes\mathbf{I}_{l_{F}}\right)\mathbf{X}\mathbf{H}\right)+\text{vec}\left(\left(\mathbf{H}_{c}^{\eta}\otimes\mathbf{I}_{l_{F}}\right)\mathbf{Z}\right)+\text{vec}(\mathbf{W})\\
 & =\left(\mathbf{H}^{T}\otimes\mathbf{H}_{c}^{\eta}\otimes\mathbf{I}_{l_{F}}\right)\mathbf{x}+\left(\mathbf{I}_{M}\otimes\mathbf{H}_{c}^{\eta}\otimes\mathbf{I}_{l_{F}}\right)\mathbf{z}+\mathbf{w}\\
 & =\mathbf{\bar{H}}_{\text{eq}}\mathbf{x}+\mathbf{\bar{z}}_{\text{eq}},
\end{aligned}
\label{eq:-22}
\end{equation}
where $\mathbf{\bar{H}}_{\text{eq}}=\mathbf{H}^{T}\otimes\mathbf{H}_{c}^{\eta}\otimes\mathbf{I}_{l_{F}}$
is the overall equivalent channel comprising both cable and radio
channels over all ENs, and $\mathbf{\bar{z}}_{\text{eq}}=\left(\mathbf{I}_{M}\otimes\mathbf{H}_{c}^{\eta}\otimes\mathbf{I}_{l_{F}}\right)\mathbf{z}+\mathbf{w}$
is the overall noise vector comprising both the vectorized radio noise
$\mathbf{z}=\text{vec}(\mathbf{Z})\sim\mathcal{CN}(\mathbf{0},\mathbf{I}_{n_{F}M})$
and the vectorized cable noise $\mathbf{w}=\text{vec}(\mathbf{W})\sim\mathcal{CN}(\mathbf{0},\frac{1}{\lambda^{2}\eta}\mathbf{I}_{n_{F}M})$,
where we recall that the scaling $\lambda$ is due to the cable power
constraints. Hence, the covariance of the equivalent noise $\bar{\mathbf{z}}_{\text{eq}}$
yields
\begin{equation}
\begin{aligned}\bar{\mathbf{R}}_{z_{\text{eq}}} & =\mathbb{E}\left[\mathbf{\bar{z}}_{\text{eq}}\mathbf{\bar{z}}_{\text{eq}}^{H}\right]\\
 & =\left(\mathbf{I}_{M}\otimes\mathbf{H}_{c}^{\eta}\otimes\mathbf{I}_{l_{F}}\right)\left(\mathbf{I}_{M}\otimes\mathbf{H}_{c}^{\eta}\otimes\mathbf{I}_{l_{F}}\right)^{H}+\frac{1}{\lambda^{2}\eta}\mathbf{I}_{n_{F}M}\\
 & =\mathbf{I}_{M}\otimes\mathbf{H}_{c}^{\eta}\mathbf{H}_{c}^{\eta}\otimes\mathbf{I}_{l_{F}}+\frac{1}{\lambda^{2}\eta}\mathbf{I}_{n_{F}M}.
\end{aligned}
\end{equation}
The eMBB per-UE rate under OMA is computed by
\begin{equation}
\begin{aligned}R_{B} & =\frac{(1-L_{U}^{-1})}{n_{F}M}I(\mathbf{r},\mathbf{x})\\
 & =\frac{(1-L_{U}^{-1})}{n_{F}M}\text{log}\left(\text{det}\left(\mathbf{I}+\bar{P}_{B}\bar{\mathbf{R}}_{\bar{z}_{\text{eq}}}^{-1}\mathbf{\bar{H}}_{\text{eq}}\mathbf{\bar{H}}_{\text{eq}}^{T}\right)\right),
\end{aligned}
\label{eq:-17}
\end{equation}
where $\bar{\mathbf{H}}_{\text{eq}}\mathbf{\bar{H}}_{\text{eq}}^{T}=\mathbf{H}\mathbf{H}\otimes\mathbf{H}_{c}^{\eta}\mathbf{H}_{c}^{\eta}\otimes\mathbf{I}_{l_{F}}$.
Finally, using the determinant property of the Kronecker product $\left|\mathbf{A}\otimes\mathbf{B}\right|=\left|\mathbf{A}\right|^{m}\left|\mathbf{B}\right|^{n},$
Equation (\ref{eq:-17}) simplifies to
\begin{equation}
\begin{aligned}R_{B} & =\frac{l_{F}(1-L_{U}^{-1})}{n_{F}M}\text{log}\left(\text{det}\left(\mathbf{I}+\bar{P}_{B}\mathbf{R}_{\bar{z}_{\text{eq}}}^{-1}\mathbf{H}_{\text{eq}}\mathbf{H}_{\text{eq}}^{T}\right)\right)\\
 & =\mu\frac{1-L_{U}^{-1}}{M}\text{log}\left(\text{det}\left(\mathbf{I}+\bar{P}_{B}\mathbf{R}_{z_{\text{eq}}}^{-1}\mathbf{H}_{\text{eq}}\mathbf{H}_{\text{eq}}^{T}\right)\right),
\end{aligned}
\label{eq:-23}
\end{equation}
where $\mathbf{H}_{\text{eq}}=\mathbf{H}\otimes\mathbf{H}_{c}^{\eta}$
and $\mathbf{R}_{z_{\text{eq}}}=\mathbf{I}_{M}\otimes\mathbf{H}_{c}^{\eta}\mathbf{H}_{c}^{\eta}+\frac{1}{\lambda^{2}\eta}\mathbf{I}_{\frac{M}{\mu}}$,
thus concluding the proof.

\section{Proof of Lemma 3\label{subsec:Proof-of-Lemma}}

By substituting the signal (\ref{eq:-14}) received at the $k$-th
EN under NOMA into (\ref{eq:-16}) and by following similar steps
as for the proof of Lemma \ref{lem:In-the-case}, the overall vector
signal received at the BBU comprising both cable and radio channels
over all ENs under NOMA by treating interference as noise yields
\begin{equation}
\mathbf{r}=\mathbf{\bar{H}}_{\text{eq}}\mathbf{x}+\beta\bar{\mathbf{A}}_{\text{eq}}\mathbf{u}+\mathbf{\bar{z}}_{\text{eq}},
\end{equation}
where the definitions of $\mathbf{\bar{H}}_{\text{eq}}$ and $\mathbf{\bar{z}}_{\text{eq}}$
are the same as in (\ref{eq:-22}), $\bar{\mathbf{A}}_{\text{eq}}=\mathbf{A}\otimes\mathbf{H}_{c}^{\eta}\otimes\mathbf{I}_{l_{F}}$
accounts for the relay of URLLC signal over the cable, and $\mathbf{u}$
is the vectorization of the URLLC signal matrix $\mathbf{U}$ in (\ref{eq:-14}).
Similarly to the proof of Lemma \ref{lem:In-the-case}, the eMBB per-UE
rate under NOMA by treating URLLC interference as noise can be computed
by
\begin{equation}
\begin{aligned}R_{B} & =\frac{1}{n_{F}M}I\left(\mathbf{r},\mathbf{x}\left|\mathbf{A}\right.\right)\\
 & =\frac{\mu}{M}\mathbb{E}_{\mathbf{A}}\left[\text{\text{log}}\left(\text{det}\left(\mathbf{I}+P_{B}\mathbf{R}_{A,z_{\text{eq}}}^{-1}\mathbf{H}_{\text{eq}}\mathbf{H}_{\text{eq}}^{T}\right)\right)\right]
\end{aligned}
\end{equation}
where the average is taken over all the possible values of matrix
$\mathbf{A}$, $P_{B}$ is the power of the eMBB user under NOMA and
$\mathbf{H}_{\text{eq}}$ is defined as in (\ref{eq:-23}). Finally,
the covariance $\mathbf{R}_{A,z_{\text{eq}}}$ of the noise plus URLLC
interference yields
\[
\begin{aligned}\mathbf{R}_{A,z_{\text{eq}}} & =\beta^{2}P_{U}\end{aligned}
\mathbf{A}_{\text{eq}}\mathbf{A}_{\text{eq}}^{T}+\mathbf{R}_{z_{\text{eq}}},
\]
where $\mathbf{A}_{\text{eq}}=\mathbf{A}\otimes\mathbf{H}_{c}^{\eta}$
and $\mathbf{R}_{z_{\text{eq}}}$ is the same as in (\ref{eq:-23}).
The proof is completed by noticing that $\mathbf{A}_{\text{eq}}\mathbf{A}_{\text{eq}}^{T}=(\mathbf{A}\otimes\mathbf{H}_{c}^{\eta})(\mathbf{A}\otimes\mathbf{H}_{c}^{\eta})^{H}=\mathbf{A}\otimes\mathbf{H}_{c}^{\eta}\mathbf{H}_{c}^{\eta}$,
since matrix $\mathbf{A}$ is idempotent and $\mathbf{H}_{c}^{\eta}$
symmetric.


\end{document}